\documentclass[10pt]{article}
\usepackage{comment}
\def\ZK{{\mathbb{Z}}_K}

\usepackage{amsmath,amsthm,amsfonts,amssymb}
\usepackage{graphicx}
\usepackage{enumerate}
\usepackage{float}
\usepackage{makecell}
\usepackage{tikz-cd}
\usetikzlibrary{arrows}
\usepackage{adjustbox}
\usepackage{extarrows}
\usepackage{enumitem}

\newtheorem{theorem}{Theorem}[section]
\newtheorem{lemma}[theorem]{Lemma}
\newtheorem{proposition}[theorem]{Proposition}

\newtheorem{remark}{Remark}

\def\Z{\mathbb{Z}}
\def\N{\mathbb{N}}
\def\Q{\mathbb{Q}}
\def\R{\mathbb{R}}
\def\lra{\ \longrightarrow\ }

\def\jj{\mathfrak{j}}
\def\mm{{\mathfrak{m}}}
\def\cK{\mathcal{C}_K}
\def\ZKx{{\Z^\times_K}}
\def\uf{u_o}

\newcommand{\lhra}{\lhook\joinrel\relbar\joinrel\rightarrow}
\newcommand{\Gal}[2]{{\mathrm{Gal}_{#1/#2}}}
\newcommand{\ag}[1]{{\ZK/{(#1)}}}
\newcommand{\mg}[1]{{\bigl(\ag{#1}\bigr)^\times}}
\newcommand{\U}[1]{{U_{#1}^1}}

\DeclareMathOperator{\coker}{coker}
\DeclareMathOperator{\nm}{{\mathbf{N}}}
\DeclareMathOperator{\img}{im}

\DeclareFontFamily{U}{rcjhbltx}{} 
\DeclareFontShape{U}{rcjhbltx}{m}{n}{<->rcjhbltx}{}
\DeclareSymbolFont{hebrewletters}{U}{rcjhbltx}{m}{n}
\DeclareMathSymbol{\shin}{\mathord}{hebrewletters}{152}

\usepackage{xcolor}
\usepackage[
  final,
  colorlinks=true,
  linkcolor=blue,
  citecolor=blue,
  urlcolor=blue,
  linktoc=all
 ]{hyperref}

\begin{document}
\synctex=1

\author{Ingemar Bengtsson*, Gary McConnell**}

\title{HOW STARK UNITS ENTER SIC OVERLAPS}

\maketitle

\begin{center}

{\small 
${}^{*}${\sl Stockholms Universitet, AlbaNova,}

{\sl SE-106 91 Stockholm, Sverige}
}

\

{\small 
${}^{**}${\sl Controlled Quantum Dynamics Theory Group, Imperial College,}

{\sl London, United Kingdom SW7 2AZ}
}

\vspace{8mm} 

{\bf Abstract:}

\end{center}

\noindent It has been observed that the mutual scalar products of the vectors in 
a SIC-POVM are given by algebraic units, and at least in some cases by square roots of Stark units. 
The full picture is somewhat more 
complicated, especially if non-minimal SIC-POVMs are considered. We present a mixture 
of exact and numerical evidence suggesting that the overlap units are always products 
of integral powers of square roots of Stark units from ray class fields all of which are attached to the maximal ring of 
integers in the base field. In the non-minimal case a lattice of such ray class fields is involved. 
In every second dimension (counted in a certain way) some of the overlap units equal $\pm 1$, and we 
show that this follows from a special property of the ray class fields. Our observations are 
complementary to but consistent with the claim that the overlap units can be calculated directly from the 
Shintani--Faddeev modular cocycle. 

\vspace{10mm} 

\tableofcontents

\section{Introduction}\label{sec:sec1}
\noindent We will be concerned with the intersection of two unsolved 
problems, one of them coming from number theory and the other from 
quantum theory. It came as a considerable surprise that this 
intersection exists \cite{AYAZ, AFMY}. 

The number theory problem is a subproblem of Hilbert's 12th: find an 
elegant description of the generators of the most general abelian 
extension of a real quadratic number field. (This formulation is a little 
vague, but so was Hilbert's \cite{Hilbert}.) A real quadratic field is 
an extension $K = {\mathbb{Q}}(\sqrt{D})$ of the rational field ${\mathbb{Q}}$, 
where $D$ is a positive square-free integer. In general an extension of a 
number field is said to be abelian if its Galois group is abelian. It has 
been known for a long time what the abelian extensions are. One fixes an 
ideal in the ring of integers of the base field, in our case $K$, and defines what is known 
as a ray class field with that ideal as the finite part of its modulus. Every abelian 
extension is a subfield of such a ray class field. 
Fifty years ago Harold Stark gave a procedure that conjecturally allows us to 
calculate algebraic units in any such ray class field \cite{Stark1}. It 
starts from an analytic zeta function attached to the field and---provided 
that the numerical calculation is carried out with enough precision---ends 
with an exact expression for what are known as Stark units. In most but not 
quite all of the cases we consider they serve as generators of their 
ray class field. Much more can be said, a main point 
being that there does not at this time exist a proof of Stark's conjectures, 
not even when the base field is~$K$. 

A SIC---the name is a shortening of the unilluminating acronym SIC-POVM---can 
be projectively defined as an orbit of the Weyl--Heisenberg 
group in the Hilbert space ${\mathbb{C}}^d$ with a certain equiangularity 
property \cite{Zauner, Renes}. We have placed the group theoretical 
details in Appendix \ref{sec:A1}, so as to not prolong this introduction 
unnecessarily for readers in the know. Once these details are accepted 
it is enough to consider the group elements $\{ D_{i,j} \}_{i,j = 0}^{d-1}$, 
where $D_{0,0}$ is the identity element. Choose a fiducial unit vector $ \lvert \Psi_0 \rangle$ 
for the group to act on. By definition the orbit is a SIC if and only if there exist 
phase factors $e^{i\theta_{j,k}}$, the indices $j$ and $k$ not both zero, such that 

\begin{equation} \langle \Psi_0 \lvert D_{0,0} \lvert \Psi_0 \rangle = 1\ , \hspace{8mm}
\langle \Psi_0 \lvert D_{j,k} \lvert \Psi_0 \rangle = \frac{e^{i\theta_{j,k}}}
{\sqrt{d+1}} \ . \label{SIC} \end{equation}

\noindent By taking the absolute values we see that the $d^2$ vectors in 
the orbit are equiangular, and it is easy to prove that more than $d^2$ 
equiangular unit vectors cannot exist in dimension $d$. The open question 
is whether this upper bound can be achieved (in this way) for every choice 
of the dimension $d$. Good reasons to be interested in SICs come from 
foundational concerns in quantum theory and from some quantum engineering 
applications \cite{Chris}, but we do not go into this here. 

We focus on the quite remarkable observation that the unspecified phase 
factors $e^{i\theta_{j,k}}$ that appear in the definition seem to be algebraic 
units as soon as $d > 3$ \cite{AFMY}. For the case $d = p$ where $p$ is a prime 
equal to 2 modulo 3 Gene Kopp established in four examples that they are indeed square 
roots of Stark units \cite{Kopp1}. There are other cases where Stark units enter the overlaps in a 
slightly different way \cite{ABGHM, BGM}. For brevity we will refer to the phase factors 
as {\it overlap units}. 

The bridge between Hilbert space and number theory is the remarkable 
formula \cite{AYAZ} 

\begin{equation} (d + 1)(d - 3) = f_0^2 \Delta_0 \ , \hspace{5mm} \Delta_0 = 
\left\{ \begin{array}{ll} D & \mbox{if $D = 1$ mod 4} \\ 
4D & \mbox{otherwise.} \end{array} \right. \label{Eq1} \end{equation}

\noindent Here $D$ is a square-free integer used to define a real 
quadratic field, $\Delta_0$ is a \emph{fundamental discriminant}, and $d$ is 
the dimension of the Hilbert space in which the SIC exists. 
In our set-up, $d$ is \emph{also} the modulus defining a ray class field with 
the quadratic field as its base field. The prevailing conjecture, verified in all 
cases examined, is that this ray class field can be used to construct 
a SIC in dimension $d$ \cite{AFMY}. The integer $f_0$ is of 
interest too: its divisors allow us to define an entire lattice of 
abelian extensions of the minimal ray class field, and to predict the 
`spectroscopy' of unitarily inequivalent SICs in that dimension 
\cite{Kopp5, Kopp2}.

Early studies of the SIC existence problem included extensive numerical 
searches, leading 
Scott and Grassl to develop a `SIC phenomenology' that tells us 
what symmetries the SICs enjoy in what dimensions \cite{Scott1, Scott2}. 
The full picture is quite intricate, but it can now be derived from 
the conjectures about what number fields are needed 
to construct the SICs. In a recent contribution Kopp has proposed a special 
function known as the Shintani--Faddeev modular cocycle \cite{Kopp3} which 
is believed to produce the algebraic units that eventually, after the 
appropriate Galois transformation, appear as SIC overlaps, for 
{\it every} SIC \cite{Kopp4}. The aim of the present paper is then 
to see exactly how these units relate to Stark units in standard ray 
class fields. 

We take it as given that for each divisor $f$ of $f_0$ in equation 
(\ref{Eq1}) there exists a SIC that can be constructed in a well defined subfield 
of a ray class field with (finite) modulus $fd$, and that there exists an algorithm 
that allows us to compute the Stark units $\mbox{S}_{fd}$ in that ray class field. 
If $f$ is itself composite, $f = f_1\cdots f_n$, there is a lattice of subfields 
each containing its own Stark units $\mbox{S}_{f_id}$. We claim that for each 
choice of $d > 3$ and $f \lvert f_0$ there exists a SIC such that: 

\begin{itemize} 

\item{Every overlap unit is a product of square roots of Stark units $\mbox{S}_{f_id}$. 
For non-minimal SICs each divisor $f_i$ of $f$ contributes its own square roots 
of Stark units to the product.} 

\item{If $d$ admits non-trivial divisors (as an integer in $K$) there will 
be a corresponding lattice of ray class subfields, giving rise to `baby overlaps' 
constructed from Stark units in these subfields.} 

\item{In some special `baby overlaps' the square rooted Stark units appear raised to 
an integer exponent. This exponent is determined by a simple rule \cite{Markuspower}, 
equation (\ref{Grasslrule}) below, applied separately to each of the subfields defined 
by the divisors of $f$.} 

\item{If either of $d+1$ or $d-3$ is a square the Stark units are 
forced to be trivial in certain subfields, giving rise to overlap units equal to $\pm 1$.}

\end{itemize}

\noindent As we shall see, these claims are not in contradiction with the standard picture 
that minimal SICs in even dimensions are constructed from ray class fields with modulus $2d$ \cite{AFMY}, 
or that non-minimal SICs are constructed using proper subfields of the ray class field with 
modulus $fd$ \cite{Kopp2} (or $2fd$, as the case may be). 

We will support the fourth of our claims by a theorem concerning the behaviour of ray class 
fields with modulus $d-2$, assuming eq. (\ref{Eq1}) is still in place. In fact we will find 
that the moduli $d-3$, $d-2$, $d-1$, and $d$, have a very special status. 

Sections \ref{sec:sec2} and \ref{sec:sec3} contain some preliminaries, and section 
\ref{sec:sec4} discusses the case of minimal SICs---which is governed by the divisors 
of the integer $d$. Section \ref{sec:sec5} gives some 
details concerning the number fields needed in the non-minimal case, and Section 
\ref{sec:sec6} describes the non-minimal overlaps---for which we also have to consider 
the divisors of the integer $f_0$. Sections \ref{sec:sec7} and \ref{sec:sec8} address two 
important special cases. 
Up to this point our paper should be regarded as nothing more than a collection of facts, 
but for the second of these special cases (that of `aligned SICs') 
the theorem proved in Section 9 provides a full explanation. 

\section{Facts about ray class fields and Stark units}\label{sec:sec2}
\noindent This section is a sketch of some of the number theory  
underlying our problem. Although we will touch on 
class field theory in Sections~\ref{sec:sec5} and~9, 
we refer the interested reader to  the various 
standard treatises \cite{Neukirch, Lang, gras, Tate}. 
For orientation, consult Figure \ref{fig:Nonmin2}. 
SICs that can be constructed using the fields 
shown there are known as {\it minimal}. 
In many dimensions there exist also {\it non-minimal} 
SICs, needing a further abelian extension of the 
minimal fields. We will come to them in Section~\ref{sec:sec5}. 

\begin{figure}[t]
\centerline{ \hbox{
		\includegraphics[width=60mm]{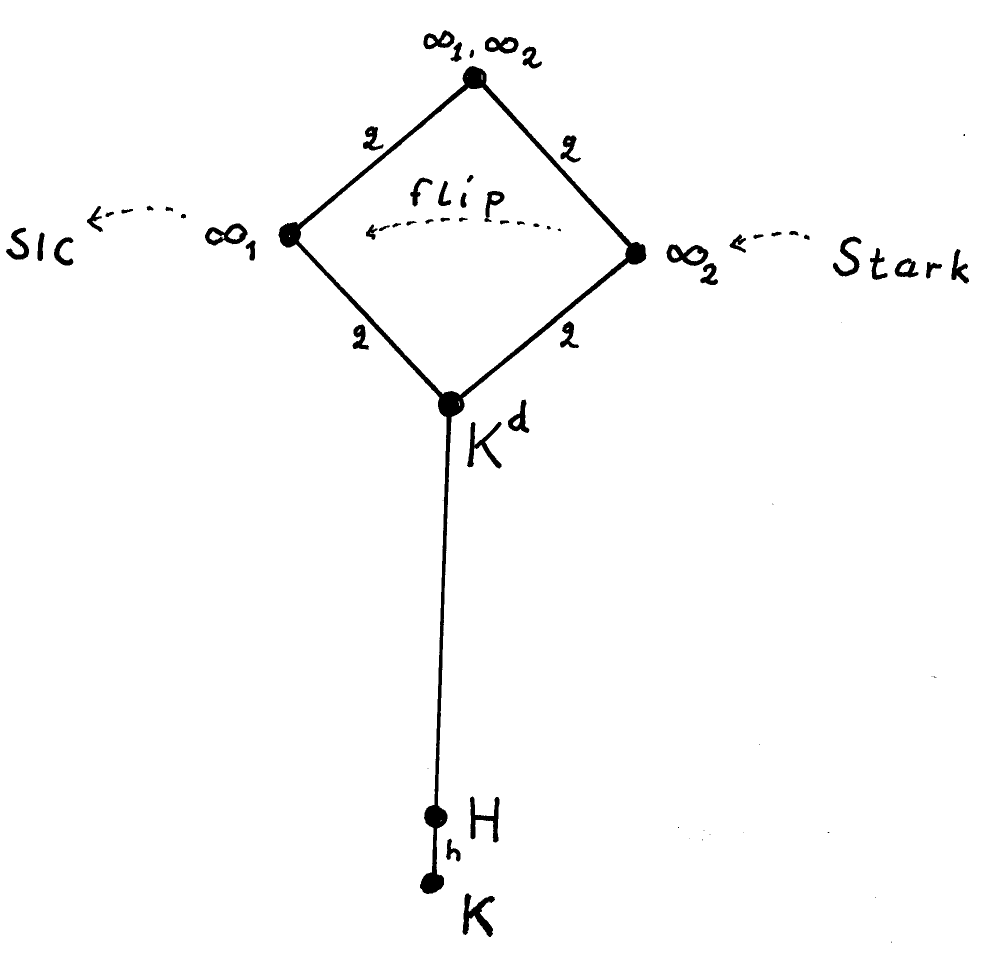} } }
\caption{\small{A field inclusion diagram including the real quadratic field $K$, 
the Hilbert class field $H$, and the four ray class fields $K^d$, $K^{d\infty_1}$, 
$K^{d\infty_2}$, and $K^{d\infty_1\infty_2}$. The degree $h_K$ of the extension 
from $K$ to $H$ determines the number of unitarily non-equivalent minimal SICs. 
To go from the real field 
$K^{d\infty_2}$ to the complex field $K^{d\infty_1}$, as required by the 
procedures that have been proposed for constructing SICs, we have to flip the 
sign of $\sqrt{D}$. This is where we rely on an exact description of the Stark units.}} 
\label{fig:Nonmin2}
\end{figure}

\paragraph{The class number $h_K$ and prime factorisation} 
At the bottom of the lattice 
sits the real quadratic field $K = {\mathbb{Q}}(\sqrt{D})$. It 
is extended to the Hilbert class field $H$. 
The Galois group Gal$(H/K)$ is important in the SIC problem because it acts on the minimal SIC 
to create a multiplet of $h_K$ unitarily inequivalent SICs \cite{AFMY}. 
Here $h_K$ is the {\it class number}, namely the order of the \emph{ideal class group}~$\cK$ 
of $\ZK$, the ring of integers in $K$. The ideal class group is isomorphic to the Galois 
group Gal$(H/K)$, so the class number is also the degree of the extension. 
If $h_K = 1$ the extension is trivial. 
If $h_K>1$ prime factorisation in $\ZK$ fails to be unique, and we may have to 
express the \emph{ideal} $p\ZK$ as a product of \emph{non-principal ideals}. 
Ideal factorisation is unique. 

Regardless of the value taken by $h_K$ it can 
happen that a rational prime $p$ does not remain prime in $\ZK$. 
If $p\ZK$ is a prime ideal in $\ZK$ then $p$ is said to be {\it inert}. 
It {\it splits} if there exist ideals $\partial$ and 
$\overline{\partial}$ such that $p\ZK = \partial\overline{\partial}$, and it {\it ramifies} 
if it is a power of a prime ideal. For the real quadratic fields $K = {\mathbb{Q}}(\sqrt{D})$, and for 
odd primes $p$, this behaviour is determined by the value of $D$ modulo $p$. 
If $D$ mod $p$ is a quadratic residue the prime splits, and if it is a 
non-residue the prime is inert. If $p$ divides into the discriminant it ramifies, 
$p\ZK = \mathfrak{p}^2$ where $\mathfrak{p}$ is a unique ideal 
(principal or non-principal) of $\ZK$. For odd primes this happens if 
and only if $p \lvert D$. 

\paragraph{Ray class fields} The key extension is that from Hilbert's class field to the 
ray class field with {\it modulus} $d$. This actually refers to the principal ideal 
$d\ZK$, but since $d$ is a rational integer our notation ignores this. When $d = 1$ 
we recover the Hilbert class field. 
To the finite modulus $d$ one can add one or two `infinite places', and 
this gives rise to the diamond at the top of Figure \ref{fig:Nonmin2}. The two infinite 
places correspond to the two different 
ways in which $\sqrt{D}$ can be embedded in the real numbers. When the 
finite part of the modulus equals $d$, as it appears in the key equation 
(\ref{Eq1}), the degree of the ray class field always increases by a factor 
of 2 when an infinite place is added \cite{AFMY}. This gives us four ray class 
fields to consider, conveniently denoted by $K^d$, $K^{d\infty_1}$, $K^{d\infty_2}$, 
and $K^{d\infty_1\infty_2}$. 

The first of these is totally real, meaning that 
regardless of how it is embedded in the complex field it consists only 
of real numbers. The next two are isomorphic in an abstract sense, but once we 
choose a definite embedding of $\sqrt{D}$ as a positive or negative real number one of them 
consists only of real numbers while the other contains genuinely complex 
numbers. 
In these specific situations where we know that one of the fields is real 
and the other is complex, we have chosen the ordering $\infty_1,\infty_2$ 
in such a way that~$K^{d\infty_2}$ is real and $K^{d\infty_1}$ complex. 
Moreover since we shall have need of actual fixed conjugate 
embeddings~$\jj \colon K\lra\R$ and~$\jj^\tau \colon K\lra\R$, we clarify the notation here. 
We denote by $\tau$ the generator of the Galois group $\Gal{K}{\Q}$. 
The embedding~$\jj$ represents the $\infty_2$ 
equivalence class in which every mapping sends 
$\sqrt{D}$ to the positive real number. 
Similarly its $\tau$-conjugate $\jj^\tau$ is in 
the $\infty_1$ equivalence class and sends $\sqrt{D}$ to a negative real number. 
This ordering was chosen to align with conventions in the Magma programme~\cite{Magma}.

The field $K^{d\infty_1\infty_2}$ is the largest of 
the four, and contains the cyclotomic field generated by the $d$th roots 
of unity as a subfield~\cite[\S4.2]{AFMY}. We 
need the roots of unity to construct the SIC because we 
need them to represent the Weyl--Heisenberg group, but it appears 
that the fiducial vector can always be chosen so that the SIC overlaps 
belong to the field $K^{d\infty_1}$ (or to the field $K^{2d \infty_1}$ if $d$ 
is even, but as we will see this is irrelevant in our context). 

The catch is that the Stark units that we calculate using Stark's 
procedure belong to the real field $K^{d\infty_2}$. If we know their minimal 
polynomial over $K$, then they can be turned into the complex units that we need by means 
of the Galois transformation $\sqrt{D} \rightarrow - \sqrt{D}$, and they will 
indeed be found to sit on the unit circle. The minimal polynomial is needed also 
to remove the sign ambiguities that arise when we take their square roots. Hence 
the numerical calculation of the real Stark units has to be performed with a 
precision high enough to identify the minimal polynomial, even if we are content 
to perform the rest of the calculation numerically with modest precision. 
 
\paragraph{Stark units} In outline, the calculation of the Stark units proceeds as follows. Let $A$ be 
a ray ideal class modulo ${\mathfrak{m}}$, where the finite part of ${\mathfrak{m}}$ 
is an ideal in $\ZK$, the ring of integers in $K$.  For us, the finite part 
will usually be a principal ideal generated by an integer from $\ZK$. Hence 
$A$ is an element in the {\it ray class group} modulo ${\mathfrak{m}}$ 
(see ref. \cite{Neukirch}, chapter VI). 
Then we define a (partial) zeta function as 

\begin{equation} \zeta (s, A) = \sum_{{\bf a} \in A}N({\bf a})^{-s} \ . 
\end{equation}

\noindent Here $N({\bf a})$ is the norm of the ideal ${\bf a}$. This function 
carries information about the real quadratic field to which it is attached. 
It has a pole at $s = 1$, but we can construct a function without poles by 
first introducing a certain involution $R$ and then defining 

\begin{equation} \delta(s,A) = \zeta (s,A) - \zeta (s,RA) \ . 
\label{involution} \end{equation}

\noindent Stark then conjectures \cite{Stark1} that one obtains algebraic units through 
the value at $s=0$ of the derivative of this expression with respect to~$s$: 

\begin{equation} \mbox{S}_\sigma = e^{\delta^\prime (0,\sigma )} \ . \end{equation}

\noindent Here an element $\sigma$ of the Galois group of the extension is 
used as a label, which is possible because of the Artin isomorphism between the 
ray class group modulo ${\mathfrak{m}}$ and the Galois group Gal$(K^{{\mathfrak{m}}}/K)$. 
In most cases these Stark units are expected to generate 
the number field in which they sit. For the actual calculation of the Stark 
units we rely on in-built Magma \cite{Magma} commands for computing Hecke $L$-functions. 
These are related to the zeta function by means of a kind of glorified Fourier 
transformation. A short description of the Magma program we use for the 
calculation can be found in refs. \cite{ABGHM, BGM}. 

According to the Stark--Tate conjectures the square root of a Stark unit belongs 
to an abelian extension of the ray class field which may or may not be a trivial 
extension \cite{Tate}. This is important here because while the overlaps belong 
to $K^{d\infty_1}$ the overlap units that we aim to calculate may or may not 
do so. This depends on whether the factor $\sqrt{d+1}$ that appears in 
equation (\ref{SIC}) belongs to that field. See ref. \cite{AFMY} for a precise 
statement. Since almost all our calculations will concern squares of overlaps 
we can largely ignore this issue, just as we can ignore the field $K^{2d \infty_1}$ 
which appears at this point if $d$ is even. 

Let us nevertheless elaborate a little: 
choose $d = 5$, say, for which $\sqrt{d+1}$ does 
not belong to the ray class field. If the actual overlaps do, it follows that the overlap 
units cannot be Stark units. They can be square roots of Stark units though, because 
in this case these square roots belong to the abelian extension $K^{d\infty_1}(\sqrt{d+1})$. 
This provides a (weak) rationale for why we need to take square roots. 

It is worth mentioning that the behaviour described in Figure \ref{fig:Nonmin2} 
holds for modulus $d$, but it is by no means universal. If the field is 
unaffected by the addition of an infinite place to the finite modulus then the 
involution that appears in equation (\ref{involution}) is absent, and 
the Stark units collapse to $+ 1$. We will encounter examples 
of this phenomenon below, and it will be the subject of Section 9. 

\paragraph{The unit group and the dimension towers} 
A key role is played by the {\it unit group} $\ZKx$; 
that is, the multiplicative group of 
units in the ring of integers $\ZK$. It has 
a torsion part $\{ 1, -1\}$ and in addition there is an 
infinite cyclic group generated 
by a {\it fundamental unit} $\uf$. We also define the totally positive unit 

\begin{equation} u_D = \frac{d_1-1 + \sqrt{(d_1+1)(d_1-3)}}{2} \ . \end{equation}

\noindent The subscript $1$ in $d_1$ refers to \eqref{dees} below. 
This unit is {totally positive} in the sense that it is positive regardless of the sign 
we assume for the square root. If $d_1- 3$ is a 
square it is the case that $u_D = \uf^2$, in all 
other cases $u_D = \uf$. 

If we fix the quadratic field by fixing the square-free 
number $D$ then the key equation (\ref{Eq1}) admits an infinite sequence of solutions 
$\{d_\ell \}_{\ell = 1}^\infty$ for $d$ \cite{AFMY}. In fact

\begin{equation}\label{dees} 
d_\ell = d_\ell(D) = u_D^\ell + u_D^{-\ell} + 1 .
\end{equation}

\noindent We refer to the resulting sequences of dimensions as {\it AFMY towers}. If we fix the 
quadratic field by setting $D = 5$, for example, then the corresponding tower begins 

\begin{equation} \{ d_\ell \}_{\ell = 1}^\infty = \{ 4, 8, 19, 48, 124, 323, 
844, 2208, 5799, 15128 , 39604, \dots \} \ . \label{tower} \end{equation}

The degree of the ray class field with modulus $d$ is 
inversely proportional to the 
multiplicative order of the image of the fundamental 
unit $\uf$ in $\ZK/d{\ZK}$, 
the integers taken modulo $d$. In our case this 
order is always divisible by 3, and this 
is related to the fact that the symmetries 
exhibited by SICs are found to be of an 
order divisible by 3; something first glimpsed by 
Zauner \cite{Zauner, Scott1, Marcus}. 
In fact the order of the unit $u_D$ modulo $d_\ell\ZK$ is $3\ell$. 
Thus the symmetries 
of the minimal SIC increase as we climb the tower, 
and using this knowledge 
the minimal SIC has been found, for example, in exact 
form in all dimensions listed in (\ref{tower}), with 
the exception of $d_8 = 2208$---which is the 
highest dimension in which a SIC has been found by means of a 
numerical search \cite{Fibonacci}---and 
$d_{10} = 15128$, where as yet no SIC has been found. 

If $d_\ell  - 3$ is a square then the fundamental unit $\uf$ has  
order $6\ell$ modulo~$d_\ell$. Then the degree of the ray class field drops by 
a factor of two, and the minimal SIC responds by having anti-unitary symmetry. 
For reasons that are not fully understood, fiducial vectors for SICs with 
anti-unitary symmetry can be constructed using square roots of Stark units in 
a `small' subfield, which means that exact SICs are available in high dimensions 
of this form, including some five digit ones \cite{ABGHM,BGM}.

\section{The lattice of divisors of $d$}\label{sec:sec3}

\noindent An important fact is that if the modulus of one ray class field 
divides that of another then the first ray class field is a subfield of 
the other (see Proposition II.4.1.1 in ref. \cite{gras}).
Hence the ray class field $K^{d\infty_1}$ will have a lattice of subfields 
corresponding to the divisor lattice of the integer $d$. This will manifest 
itself in the SIC through the occurrence of {\it baby overlaps}, 
that is to say some subsets of overlaps that belong to the various subfields. 
In this section we will lay out how the same pattern of baby overlaps arises 
from a consideration of the Weyl--Heisenberg group in dimension $d$, 
assuming that the SIC is left invariant by a symmetry of order three. The 
reader may find the argument somewhat loose. At the end of the section we 
will discuss one possible proof strategy.

\begin{figure}
\centerline{ \hbox{
		\includegraphics[width=55mm]{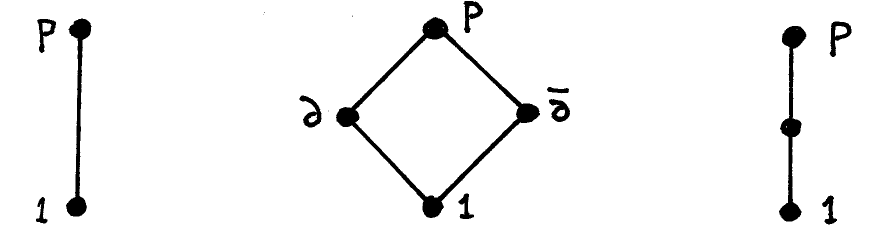} } }
\caption{\small{The lattice of divisors of a rational prime $p$ in a 
quadratic field, depending on whether $p$ is inert, splits, or ramifies.}} 
\label{fig:gitterp1}
\end{figure}

\begin{figure}
\centerline{ \hbox{
		\includegraphics[width=65mm]{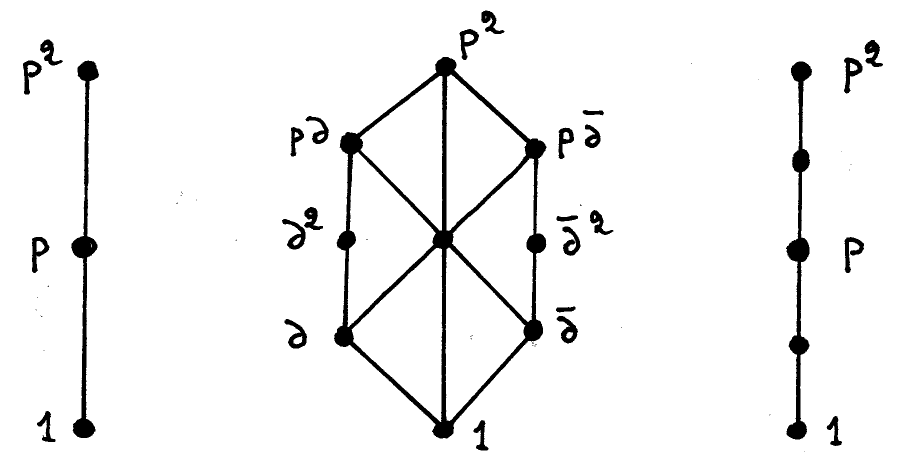} } }
\caption{\small{For $d=p^2$: the lattice of divisors of $p^2$, depending 
on whether $p$ is inert, splits, or ramifies.}} 
\label{fig:gitter22}
\end{figure}
 
We can obtain non-trivial divisor lattices already in prime dimensions $d = p$, 
because as noted in section \ref{sec:sec2} a rational prime $p$ may be composite 
considered as an integer in the quadratic field $K = \mathbb{Q}(\sqrt{D})$. If $D$ is a 
quadratic residue mod $p$ then $p$ splits, that is to say that there exist ideals $\partial$ 
and $\bar{\partial}$ such that $p =\partial \bar{\partial}$, and if $p$ divides the 
discriminant then $p$ ramifies, that is is to say it is a power of a prime ideal. 
For now we are only interested in the behaviour of primes that divide $d$, 
and the discriminant is determined by $d$ through the key formula (\Ref{Eq1}). 
We observe that 

\begin{equation} p \lvert d \hspace{5mm} \Rightarrow \hspace{5mm} 
f_0^2\Delta_0 = (d+1)(d-3) \equiv -3 \ \mbox{mod} \ p \ . \end{equation}

\noindent Let $p$ be an odd prime, so that 
$p$ divides $\Delta_0$ if and only if $p$ divides $D$. 
Hence $D$ will be a quadratic residue mod $p$ if and only if $-3$ 
is. But from quadratic reciprocity we know 
that $-3$ mod $p$ is a quadratic residue if and only if $p = 1$ mod 3. 
As long as $f_0$ and $d$ are relatively prime (this can fail only if $p = 3$) 
we conclude that $p$ 
is inert if $p=2$ mod 3, it splits if $p=1$ mod 3, and it ramifies if $p = 3$. 
The resulting divisor lattices are shown in Figure \ref{fig:gitterp1} for 
prime dimensions, and in Figure \ref{fig:gitter22} for prime-squared dimensions. 
If $d$ has several prime factors the divisor lattices become more complicated, 
but they can easily be worked out case by case. 

\begin{figure}
        \centerline{ \hbox{
        		\includegraphics[width=50mm]{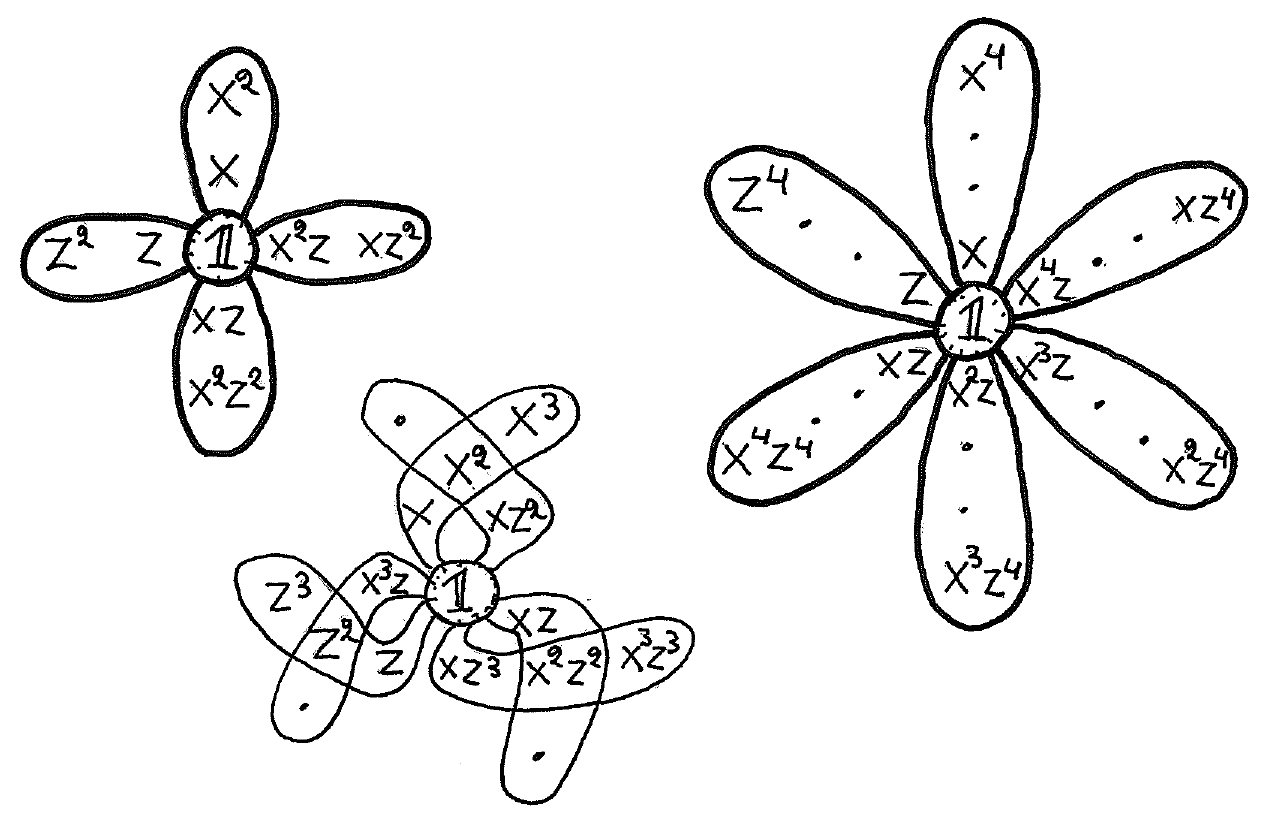} } }
        \caption{\small When the dimension is a prime $p$, 
        the Weyl--Heisenberg group 
        consists of $p+1$ cyclic subgroups 
        having only the unit element in common. The result is a flower with 
        $p+1$ petals. When $d$ is a prime power the petals of the flower intertwine 
				in a characteristic fashion, as illustrated here for $d = 2^2$.}
        \label{fig:blomma1}
\end{figure} 

Now we turn to the Weyl--Heisenberg group. Necessary definitions can be found 
in Appendix \ref{sec:A1}. The first observation to be made here is that in composite dimensions $d$ 
the group is a direct product of the Weyl--Heisenberg group in the prime power factors, 

\begin{equation} H(d) = H(p_1^{k_1}p_2^{k_2} \dots p_n^{k_n}) = 
H(p_1^{k_1})\times H(p_2^{k_2}) \times \dots \times H(p_n^{k_n}) 
\ . \end{equation}

\noindent This mirrors how prime power factors of $d$ give rise to subfields of $K^d$. Concretely, 
prime power dimensions will serve as building blocks for composite dimensions, 
hence we focus on prime power dimensions alone. It is helpful to be able to 
visualize how the maximal abelian subgroups sit inside $H(p^k)$. In prime dimensions 
$d = p$ the group can be divided into $p+1$ cyclic subgroups having only the 
identity element in common. Let us refer to them as {\it petals} of a {\it flower}. 
See Figure \ref{fig:blomma1}. In prime power dimensions the cyclic subgroups become 
entwined with each other in a characteristic fashion. The relevant group theory is 
explained elsewhere \cite{BGM, monomial}; here we just refer to Figure 
\ref{fig:blomma1} which illustrates the case $d = 2^2$. In general, if 
$d = p^2$ there are $p+1$ entwined petals formed from $p$ petals each. Each 
such entwined petal contains $(p-1)(p^2 +1)$ operators, including $p-1$ operators 
of order $p$. For higher powers of $p$ the counting can be worked 
out with modest effort.

We refer to elements of $H(d)$ as {\it displacement operators} $D_{i,j}$, 
and recall that the symplectic group acts on the displacement operators according 
to equation (\ref{symplektiskt}) in Appendix \ref{sec:A1}. If we fix the representation 
so that the group element $Z = D_{0,1}$ is diagonal then any symplectic operator that 
transforms the petal generated by $Z$ into itself will be represented by a 
monomial matrix. If in addition it transforms the petal generated by $X = D_{1,0}$ into 
itself it will be represented by a permutation matrix. Otherwise it will be 
represented by a complex Hadamard matrix whose entries are roots of unity 
\cite{Marcus}. Our key assumption is that the SIC has {\it Zauner symmetry}, 
that is to say that it is invariant under a unitary transformation representing a 
symplectic matrix of order three and trace equal to $-1$ mod $d$  \cite{Zauner, Marcus}. 
A possible choice that works in every dimension is 

\begin{equation} {\cal Z} = \left( \begin{array}{cc} 0 & - 1 \\ 1 & - 1 
\end{array} \right) \ . \end{equation}

\noindent With a suitable choice of the fiducial vector the SIC is invariant 
under this transformation in the sense that 

\begin{equation} \langle \Psi_0 \lvert U_{\cal Z}D_{i,j}U_{\cal Z}^{-1} \lvert \Psi_0\rangle 
= \langle \Psi_0 \lvert D_{-j,i-j} \lvert \Psi_0\rangle = \langle \Psi_0 \lvert D_{i,j} \lvert \Psi_0 
\rangle \ . \end{equation}

\noindent If the position $\ell$ of the dimension in the AFMY towers is higher 
than one the symmetry may be larger, and if the dimension is of the form $d = n^2+3$ the symmetry 
may include an anti-unitary symmetry as well. A different complication occurs if the 
dimension $d>3$ is divisible by 3 
but not by 9, because then there exists another conjugacy class of Zauner unitaries. 
In dimensions of the form $d = 3(3n+1)$ this leads to SICs having a symmetry 
of `type $F_a$' \cite{Scott1, Scott2}. They will be discussed in Section \ref{sec:sec7}. To keep things 
simple we momentarily ignore all special cases, and we also focus on $d = p$. 
It will then be observed that Zauner symmetry divides the $p+1$ petals of a 
flower into $k$ triplets and zero, two, or one, singlets, depending respectively upon whether $p \equiv 2, 1 \textrm{ or } 0 \bmod 3$. 

For the latter two cases it is then convenient to choose a different representative 
of the conjugacy class of order three matrices. For $p = 1$ mod 3 and $p = 3$, 
respectively, we choose $\alpha\in\Z$ with $2\leq\alpha\leq p-2$ and $\gamma \in \{1,2\}$ such that 

\begin{equation} {\cal Z}_1 = \left( \begin{array}{cc} \alpha & 0 \\ 0 & \alpha^{-1} 
\end{array} \right) \ \ \mbox{respectively} \ \ {\cal Z}_2 = 
\left( \begin{array}{cc} 1 & 0 \\ \gamma & 1 \end{array} \right) \ , 
\end{equation}

\noindent where $\alpha \neq 1$ satisfies~$\alpha^3 \equiv 1 \bmod p$ 
(such integers exist if and only if $p \equiv 1 \bmod 3$). 
The unitary transformation 
$U_{{\cal Z}_1}$ transforms the petals 
generated by $Z$ and $X$ into themselves, while $U_{{\cal Z}_2}$ transforms the 
petal generated by $Z$ into itself. Hence $U_{{\cal Z}_1}$ is a 
permutation matrix, and $U_{{\cal Z}_2}$ is at least monomial. Choosing a Zauner 
matrix of the form ${\cal Z}_1$ means that the special petals in the flower are 
those generated by $X$ and $Z$, and we always make this choice if 
$d \equiv 1 \bmod 3$. Overlap units of the form 

\begin{equation} \sqrt{d+1}\langle \Psi_0 \lvert Z^j \lvert \Psi_0\rangle = e^{i\theta_{0,j}} 
\label{Zoverlap} \end{equation}

\noindent will be referred to as {\it $Z$-overlaps}. $X$-overlaps 
are similarly defined. They will, in fact, become 
baby overlaps if $d \equiv 1 \bmod 3$. 

We are now ready to link the unitary geometry to the arithmetic in a more 
definite fashion. Let $S$ denote the symplectic symmetry group of the SIC, 
and let $M(S)$ be the 
centraliser of $S$ within the group $GL(2, {\mathbb{Z}}/d{\mathbb{Z}})$. 
The claim, made by Appleby et al. \cite{AFMY}, is that 

\begin{equation} \mbox{Gal}(K^{d\infty_1}/H) \simeq M(S)/S \ . \end{equation}

\noindent The Galois group acts on the overlaps, while the group $M(S)/S$ 
acts on the displacement operators according to the recipe in Appendix \ref{sec:A1}. 
For $d = p \equiv 2 \bmod 3$ it turns out that the action of these groups is transitive. 
For $d = p \equiv 1 \bmod 3$ the group transforms the petals generated by $X$ and $Z$ into 
themselves (assuming that we have arranged the fiducial vector suitably). The corresponding 
overlaps must then sit in proper subfields of the full field. If, say, $d = p^2$ and $d = p \equiv 1 \bmod 3$ there will be 
several short orbits of displacement operators, consisting of the subgroups generated by 

\begin{equation} \{ Z^p\}, \ \{ Z^p\}, \ \{ Z\}, \ \{ X^p, Z^p\}, \ \{ X\}, 
\ \{ X^p, Z\}, \ \{ X, Z^p\} \ . \end{equation}

\noindent Each such subgroup corresponds to a subfield, and indeed to an element of the second 
divisor lattice shown in Figure \ref{fig:gitter22}. 

The claim then is that every element in the divisor lattice for $d$ 
corresponds to baby overlaps sitting in a subfield of $K^{d\infty_1}$. 
Moreover, if baby overlaps occur it is possible to arrange the fiducial 
vector so that the $Z$-overlaps, and usually also the $X$-overlaps, 
are baby overlaps. We also claim that all the overlap units are (products 
of) square roots of Stark units raised to the power $n$, where the 
exponent is determined by the rule 

\begin{equation} n\times (\mbox{number of overlaps in the orbit}) = 
 \lvert S \lvert  \lvert \mbox{Gal} \lvert  \ , \label{Grasslrule} \end{equation}

\noindent and where in turn the Galois group is 
that of the appropriate subfield (over $H$). 
This rule was suggested by Markus Grassl \cite{Markuspower}. 

Can any of this be proved from first principles? This is a difficult matter, 
because at the time of writing SIC existence has been proved only in those 
dimensions where a SIC has been explicitly calculated. Hence we have to 
rely on some conjecture for a conditional proof. A particularly favourable case 
concerns SICs with anti-unitary symmetry, which occur 
whenever $d$ is of the form~$n^2+3$ for some $n\geq1$. 
In this case there exists a conjectural `formula' for a SIC fiducial 
vector constructed from square roots of Stark units in a proper subfield of 
$K^{d\infty_1}$ \cite{ABGHM, BGM}. The formula has been tested in close to a 
hundred cases, including (with appropriate modifications) some non-minimal 
SICs. The formula immediately implies that all the $Z$-overlaps 
are trivial, and that all the $X$-overlaps lie in the subfield containing 
the fiducial vector. With sufficient attention to detail the formula 
in fact allows one to deduce what subfields hold which baby overlaps, for 
all of the latter \cite{BGM}. Proving that the various overlaps are 
indeed given by Stark units is more difficult. It has not been done \cite{BG}. Nor has 
it been explained why we assumed a symmetry of order three in the first 
place. A comment, but no answer, can be found in Appendix \ref{sec:A2}.  

Hence, the situation is currently still far from understood. 
But the evidence to be reported below supports the claims made above.

\section{Overlaps for minimal SICs}\label{sec:sec4}

\noindent In this section we will analyse precisely how Stark units give the overlaps for 
minimal SICs, simply by giving the results for a set of selected examples. 
We begin with some known facts about prime dimensions. Suppose first that $d = p \equiv 2$ mod 3. 
Then there are no non-trivial 
divisors of $d$. Hence we are dealing with a single ray class field, and the 
order of its Galois group is~(see \cite{AFMY})

\begin{equation}  \lvert \mbox{Gal} \lvert  = 
 \lvert \mbox{Gal}(K^{d\infty_1}/H) \lvert  = \frac{p^2-1}{3\ell} \ . \end{equation}

\noindent The symmetry of the minimal SIC is of order $3\ell$ where $\ell$ is the 
position of $d$ in its AFMY tower, so the order of the Galois group equals the number 
of distinct overlaps in the SIC. We will find that all the 
overlap units are square roots of Stark units in $K^{d\infty_1}$. 

If $d = p \equiv 1$ mod 3 the prime splits over $K$, $p = \partial \overline{\partial}$, 
and it is seen in Figure \ref{fig:gitterp1} that we have three distinct Galois 
groups to consider. The details turn out to depend on whether $p-3$ is a square 
or not. If $p-3$ is not a square the orders of the Galois groups are 

\begin{eqnarray}  \lvert \mbox{Gal}(K^{d\infty_1}/H) \lvert  = \frac{(p-1)^2}{3\ell} 
\hspace{13mm} \nonumber \\ \\ 
 \lvert \mbox{Gal}(K^{\partial \infty_1}/H) \lvert  = 
 \lvert \mbox{Gal}(K^{\overline{\partial} \infty_1}/H) \lvert  = \frac{p-1}{3\ell} \ . 
\nonumber \end{eqnarray}

\noindent The symmetry is again of order $3\ell$. There will be three distinct 
Galois orbits of overlaps, and we will arrange the fiducial vector so that 
the two `small' orbits correspond to the cyclic subgroups of the Heisenberg 
group that are generated by $X$ and by $Z$. There are $p-1$ $Z$-overlaps and 
equally many $X$-overlaps, and these overlap units are found to be square 
roots of Stark units in the relevant subfield of the full ray class field. 

If $d = p$ is a prime of the form~$n^2 + 3$ the situation changes. 
We have \cite{ABGHM} 

\begin{eqnarray}  \lvert \mbox{Gal}(K^{d\infty_1}/H) \lvert  = \frac{(p-1)^2}{6\ell} \ , 
\hspace{8mm}  \lvert \mbox{Gal}(K^{\partial \infty_1}/H) \lvert  = \frac{p-1}{3\ell} \ , \nonumber \\ 
\label{from16} \\  
K^{\overline{\partial} \infty_1} = K^{\overline{\partial}} \ . \hspace{33mm} \nonumber \end{eqnarray}

\noindent The fact that $K^{\overline{\partial} \infty_1} = K^{\overline{\partial}}$ 
has the consequence that Stark's construction \cite{Stark1} applied to this 
subfield results in units that are simply equal to $+ 1$. We arrange the 
fiducial vector so that these trivial units occur as $Z$-overlaps, while the 
$X$-overlaps belong to $K^{\partial \infty_1}$. 

The minimal SIC has anti-unitary symmetry whenever $d = n^2+3$, not necessarily 
a prime, and the order of the symmetry group~$S$ is~$ \lvert S \lvert  = 6\ell$. An important point is that 
it can be proved, using pure Hilbert space arguments, that a SIC with 
anti-unitary symmetry must have $d-1$ trivial baby overlap units, equal to  $+ 1$ in odd 
dimensions  \cite{ABGHM,BGM}. With a suitable choice of the fiducial vector 
they are precisely the $Z$-overlaps. 

{\small 
\begin{table}[h]
\caption{{\small Minimal SIC overlaps in prime dimensions. For $d = p \equiv 2$ 
mod 3 the examples are $d = 5$, 11, 17, 23, 53, for $d = p = \partial 
\overline{\partial} \equiv 1$ mod 3 they are $d = 13$, 31 when $p-3$ is not a 
square. For $d = p$ of the form $n^2+3$, exact SICs were computed in thirteen cases 
\cite{ABGHM}, although the `large' Stark units S$_p$ were computed only for 
$d = 7$, 19, 67.}}
  \smallskip \smallskip
\hskip 1.0cm
{\renewcommand{\arraystretch}{1.6}
\begin{tabular}
{|c|c|c|c|c|c|}\hline \
Dimension & Type & \#  & $ \lvert \mbox{Gal} \lvert $ & $ \lvert S \lvert $ & Overlap \\
\hline \hline 
$d = p = 2$ mod 3 & $D$ & $p^2-1$ & $\frac{p^2-1}{3\ell}$ & $3\ell$ & 
$(\mbox{S}_p)^\frac{1}{2}$ \\ 
\hline \hline 
$d = p = 1$ mod 3 & $Z$ & $p-1$ & $\frac{p-1}{3\ell}$ & $3\ell$ 
& $(\mbox{S}_{\overline{\partial}})^\frac{1}{2}$ \\ 
$d \neq n^2+3$ & $X$ & $p-1$ & $\frac{p-1}{3\ell}$ & $3\ell$ 
& $(\mbox{S}_{\partial})^\frac{1}{2}$ \\ 
& $D$ & $(p-1)^2$ & $\frac{(p-1)^2}{3\ell}$ & $3\ell$ & $(\mbox{S}_p)^\frac{1}{2}$ \\ 
\hline \hline 
$d = p = 1$ mod 3 & $Z$ & $p-1$ & -- & -- & 1 \\ 
$d = n^2+3$ & $X$ & $p-1$ & $\frac{p-1}{3\ell}$ & $6\ell$ & $\mbox{S}_{\partial} $ \\ 
& $D$ & $(p-1)^2$ & $\frac{(p-1)^2}{6\ell}$ & $6\ell$ & $(\mbox{S}_p)^\frac{1}{2}$ \\ 
\hline 
\end{tabular}
}
\label{tab:primes}
\end{table} 
}

With this preamble we can present our evidence that SICs really behave 
as stated. See Table \ref{tab:primes}. This is the first of many tables 
where $D$ stands for a generic displacement operator, $X$ and $Z$ occur 
separately if they give rise to baby overlaps, $\#$ is the total number of overlaps 
accounted for, $ \lvert \mbox{Gal} \lvert $ stands for the order of the Galois group 
Gal$(K^{{\mathfrak{m}}_0\infty_1}/H)$ for the relevant finite modulus ${\mathfrak{m}}_0$ (at least 
for now---this has to be modified a little when we come to the non-minimal 
SICs), $ \lvert S \lvert $ is the order of the symmetry group, and S$_{{\mathfrak{m}}_0}$ is a 
Stark unit in the ray class field indicated by the subscript. 
With this information in hand it can 
be checked by inspection that Grassl's rule (\ref{Grasslrule}) correctly 
predicts the exponent of the square rooted Stark units that give 
the overlaps. When $d = p \equiv 1$ mod 3 the $X$-overlaps belong to 
$K^{\partial \infty_1}$, and the rule tells us that the square rooted Stark 
units should appear raised to the power $n=2$, 

\begin{equation} n\times (p-1)=6\ell \times (p-1)/3\ell \hspace{5mm} 
\Rightarrow \hspace{5mm} n = 2 \ . \end{equation}

\noindent For the `generic' or `large' overlaps we find $n = 1$ so we 
expect them to be square roots of Stark units in $K^{d\infty_1}$, and indeed 
they are. The extent to which the claims of a 
given table has been verified is stated in its caption. 

There are no surprises in prime power dimensions $d = p^k$, at least not when 
$p > 3$. We expect that $d \neq n^2+3$ and $\ell = 1$ for all these dimensions, 
where $\ell$ is the position in the AFMY tower. If so the order of the 
symmetry group is always $ \lvert S \lvert  = 3$. We give 
the pattern of unit overlaps for $d = p^2$ in Table 
\ref{tab:prim2}. It faithfully mirrors the lattice of divisors shown in Figure 
\ref{fig:gitter22}. When $p = 3$ the prime 3 ramifies, $(3) = {\mathfrak{p}}^2$, 
and we encounter a `short' orbit for which Grassl's rule correctly gives the 
power $3$ for the square rooted Stark units. 

{\small 
\begin{table}[h]
\caption{{\small Minimal SIC overlaps for $d = p^2$. The 
examples are $d = 25$, 49, 9.}}
 \smallskip \smallskip
\hskip 0.8cm
{\renewcommand{\arraystretch}{1.4}
\begin{tabular}{|c|c|c|c|c|c|}
\hline 
Dimension & Type & $\#$ & $ \lvert \mbox{Gal} \lvert $ & $ \lvert S \lvert $ & Overlap \\
\hline \hline
$d = p^2$ & $D^p$ & $p^2 - 1$ & $\frac{(p^2-1)}{3}$ & $3$ 
& $(\mbox{S}_p)^\frac{1}{2}$ \\ 
$p = 2$ mod 3 & $D$ & $(p^2-1)p^2$ & $\frac{(p^2-1)p^2}{3}$ & $3$ 
& $(\mbox{S}_d)^\frac{1}{2}$ \\ 
\hline \hline  
$d = p^2$ & $Z^p$ & $p - 1$ & $\frac{p-1}{3}$ & $3$ 
& $(\mbox{S}_{\overline{\partial}})^\frac{1}{2}$ \\ 
$p=1$ mod 3 & $Z$ & $(p - 1)p$ & $\frac{(p-1)p}{p}$ & $3$ 
& $(\mbox{S}_{\overline{\partial}\overline{\partial}})^\frac{1}{2}$ \\ 
 & $X^pZ$ & $(p - 1)^2p$ & $\frac{(p-1)^2p}{3}$ & $3$ 
& $(\mbox{S}_{p\overline{\partial}})^\frac{1}{2}$ \\ 
 & $X^p$ & $p - 1$ & $\frac{p-1}{3}$ & $3$ 
& $(\mbox{S}_{\partial})^\frac{1}{2}$ \\ 
 & $X$ & $(p - 1)p$ & $\frac{(p-1)p}{3}$ & $3$ 
& $(\mbox{S}_{\partial \partial})^\frac{1}{2}$ \\ 
 & $XZ^p$ & $(p - 1)^2p$ & $\frac{(p-1)^2p}{3}$ & $3$ 
& $(\mbox{S}_{p\partial})^\frac{1}{2}$ \\ 
 & $X^pZ^p$ & $(p - 1)^2$ & $\frac{(p-1)^2}{3}$ & $3$ 
& $(\mbox{S}_p)^\frac{1}{2}$ \\ 
 & $D$ & $(p - 1)^2p^2$ & $\frac{(p-1)^2p^2}{3}$ & $3$ 
& $(\mbox{S}_d)^\frac{1}{2}$ \\ 
\hline \hline 
$d = 3^2$ & $Z^3$ & 2 & 2 & 3 & $(\mbox{S}_3)^\frac{3}{2}$ \\ 
 & $D^3$ & 6 & 2 & 3 & $(\mbox{S}_3)^\frac{1}{2}$ \\ 
 & $Z$ & 18 & 6 & 3 & $(\mbox{S}_{3{\mathfrak{p}}})^\frac{1}{2}$ \\ 
 & $D$ & 54 & 18 & 3 & $(\mbox{S}_d)^\frac{1}{2}$ \\ 
\hline
\end{tabular}
}
\label{tab:prim2}
\end{table}
}

When the dimension is even we need to extend the centre of the 
Weyl--Heisenberg group so that it has order $2d$ in order to obtain a 
clean description of the SICs \cite{Marcus}. But this is irrelevant for 
our present purposes. The Stark units that appear in the overlaps are 
those coming from the ray class field with finite modulus $d$, together with 
its subfields. The extension to $2d$ is needed only to take the relevant 
square roots. The case $d = 2^2$ actually fits into Table \ref{tab:prim2} 
if we take into account that $K^{2,\infty_1} = K^2$ so that the corresponding 
Stark units are equal to $+1$. We give examples of the form $d = 2p$ in Table 
\ref{tab:jamna}. Note that $d_\ell = 2p$ implies $\ell = 1$, so $ \lvert S \lvert  = 3$ 
for all such dimensions \cite{BGM}. This is fortunate because the second row of 
the Table would not make sense otherwise. For $d = 6$ we see again that a 
factor 3 in the dimension results in overlaps given by square rooted Stark units raised to 
the power 3 because one of the orbits is `short'. 

{\small 
\begin{table}[h]
\caption{{\small Non-trivial minimal SIC overlaps in twice odd dimensions. The 
Hilbert space is ${\mathbb{C}}^2\otimes {\mathbb{C}}^p$. The operators $D\otimes 
{\bf 1}$ give trivial overlaps and are not listed. The 
examples are $d = 10$, 6. There are no surprises in $d = 14$ which is 
not shown.}}
 \smallskip \smallskip
\hskip 0.9cm
{\renewcommand{\arraystretch}{1.4}
\begin{tabular}{|c|c|c|c|c|c|}
\hline 
Dimension & Type & $\#$ & $ \lvert \mbox{Gal} \lvert $ & $ \lvert S \lvert $ & Overlap \\
\hline \hline
$d = 2p$ & ${\bf 1}\otimes D$ & $p^2 - 1$ & $(p^2-1)/3$ & $3$ 
& $(\mbox{S}_p)^\frac{1}{2}$ \\ 
$p = 2$ mod 3  & $D\otimes D$ & $3(p^2 - 1)$ & $p^2-1$ & $3$ 
& $(\mbox{S}_d)^\frac{1}{2}$ \\
\hline 
$d = 6$ & ${\bf 1}\otimes Z$ & 2 & 2 & 3 & $(\mbox{S}_3)^\frac{3}{2}$ \\ 
 & ${\bf 1}\otimes D$ & 6 & 2 & 3 & $(\mbox{S}_3)^\frac{1}{2}$ \\ 
 & $D\otimes Z$ & 6 & 2 & 3 & $(\mbox{S}_{\mathfrak{p}} )^\frac{1}{2}$ \\ 
 & $D\otimes D$ & 18 & 6 & 3 & $(\mbox{S}_d )^\frac{1}{2}$ \\ 
\hline
\end{tabular}
}
\label{tab:jamna}
\end{table}
}

{\small 
\begin{table}[h]
\caption{{\small Non-trivial minimal overlaps for $d = n^2 + 3 = 4p$, $p = 1$ 
mod 3. The Hilbert space is ${\mathbb{C}}^4\otimes {\mathbb{C}}^p$. 
The examples are $d = 28$, 52, 124, 172, 292, 844. This includes cases with 
$\ell = 5$, 7. For the last three dimensions only the first five rows were 
checked.}} \smallskip \smallskip
\hskip 1.5cm
{\renewcommand{\arraystretch}{1.0}
\begin{tabular}
{|c|c|c|c|c|}\hline \ 
Type & \# & $ \lvert \mbox{Gal} \lvert $ & $ \lvert S \lvert $ & Overlap \\
\hline \hline
${\bf 1} \otimes X$ & $p-1$ & $(p-1)/3\ell$ & $6\ell$ & $\mbox{S}_\partial $ \\
$D_{0,2}\otimes X$ & $3(p-1)$ & $(p-1)/\ell$ & $6\ell$ & $\mbox{S}_{2\partial}$ \\
${\bf 1}\otimes D$ & $(p-1)^2$ & $(p-1)^2/6\ell$ & $6\ell$ & $(\mbox{S}_p)^\frac{1}{2}$ \\ 
$D_{0,2} \otimes D$ & $3(p-1)^2$ & $(p-1)^2/2\ell$ & $6\ell$ & $(\mbox{S}_{2p})^\frac{1}{2}$ \\
$D\otimes {\bf 1}$ & 12 & 2 & $6\ell$ & $(\mbox{S}_4)^\frac{\ell}{2}$ \\ 
$D\otimes Z$ & $12(p-1)$ & $2(p-1)/\ell$ & $6\ell$ & 
$(\mbox{S}_{4\overline{\partial}})^\frac{1}{2}$ \\
$D \otimes X$ & $12(p-1)$ & $2(p-1)/\ell$ & $6\ell$ & $(\mbox{S}_{4\partial})^\frac{1}{2}$ \\
$D \otimes D$ & $12(p-1)^2$ & $2(p-1)^2/\ell$ & $6\ell$ & $(\mbox{S}_{d})^\frac{1}{2}$ \\
\hline 
\end{tabular}
}
\label{tab:4pfaser}
\end{table}
}

For generic composite dimensions the lattice of divisors becomes more 
complicated without becoming more interesting. There are a few things to 
notice. First, if $d = p_1p_2$ with neither of the primes being of the form 
$n^2+3$ it can still happen that their product is. Then some trivial overlaps 
occur in the product dimension, an example being $403 = 13\cdot 31$. Conversely, 
$7\cdot 19 = 133$, and there are no trivial overlap phases there (or should not 
be---we have not checked). Two very special cases are dimensions of the 
form $d = 3(3n+1)$ or of the form $d = n(n-2)$. They will be discussed 
separately, in  Sections \ref{sec:sec7} and \ref{sec:sec8}. 

A comment on the calculations is in order. For the examples in $d = p \equiv 2$ mod 3 
we constructed the exact SIC using Kopp's method \cite{Kopp1}. This entails 
calculating the Stark units numerically to enough precision so that we can 
determine their minimal polynomial. With this in hand we can perform the 
`flip' indicated in Figure \ref{fig:Nonmin2}, and take the global square root 
if necessary. 
We then have to find an exact root of the resulting polynomial, act on this root 
with the Galois group, and place the resulting numbers in the correct order as 
overlap units. This will allow us to construct the SIC using the Schwinger 
formula (\ref{Schwinger}), 
and check with an exact calculation that the resulting 
expression really is a projector of rank 1. The details can be found in Kopp's 
paper. For $d = 53$ the calculation is made easier by the fact that the symmetry 
is of order nine ($\ell = 3$), but it nevertheless took 4 hours to calculate 
the Stark units and 13 hours to find the exact root that we needed (using 
Magma). For most of the remaining examples we took a more pragmatic approach. 
We calculated the minimal polynomial of the complex Stark units in exact 
form, and then checked numerically to high precision that the overlap units 
raised to the suitable power provide its roots. 

When the dimension is of the form $d = n^2+3$ we can again resort to exact 
calculations, and give many more examples. For prime dimensions of this form 
the results were given in Table \ref{tab:primes}. For the case $d = n^2+3 = 4p$ 
see Table \ref{tab:4pfaser}, which is taken from ref. \cite{BGM}. Overlap units 
equal to $\pm 1$ were not included in this table. Note that $\ell$, 
the position in the AFMY towers, occurs in the exponents of the Stark units in one 
of the rows, as predicted by the rule (\ref{Grasslrule}). 

\section{The lattice of divisors of $f_0$}\label{sec:sec5}

We now turn to the non-minimal SICs. Our first task 
is to explain the role of the integer $f_0$ that appears in 
the formula (\ref{Eq1}). Let $f$ be any divisor of $f_0$. Kopp and 
Lagarias use these integers to define subrings of the ring of 
integers in the quadratic field, and to define a new kind of ray class 
field which depends not only on the modulus $d$ but also on the 
integer $f$ \cite{Kopp5}. They refer to $f$ as the {\it conductor} (namely 
of a non-maximal ring of integers in the quadratic field). Our homemade 
notation for these new ray class fields is $K_{L}^{d; f}$. Clearly 
$K_{L}^{d;1} = K^d$, and as usual the 
degree of $K_{L}^{d\infty_1; f}$ is twice that of the totally real field 
$K_{L}^{d;f}$. It turns out that $K_{L}^{d;f}$ is always a subfield of 
$K^{fd}$. We will be interested in whether it is a proper subfield, 
or not. For this reason we define the {\it excess} $e$ as the index of the 
extension, 

\begin{equation} \mbox{excess} = e = [ K^{fd}: K^{d;f}_L ] = 
\frac{\mbox{degree}(K^{fd})}{\mbox{degree}
(K_{L}^{d;f})} \ . \label{excess} \end{equation}

\noindent If the excess equals one, then the two number fields are identical 
(this is also the case if we add one or two infinite places to the moduli). 

\begin{figure}
\centerline{ \hbox{
		\includegraphics[width=30mm]{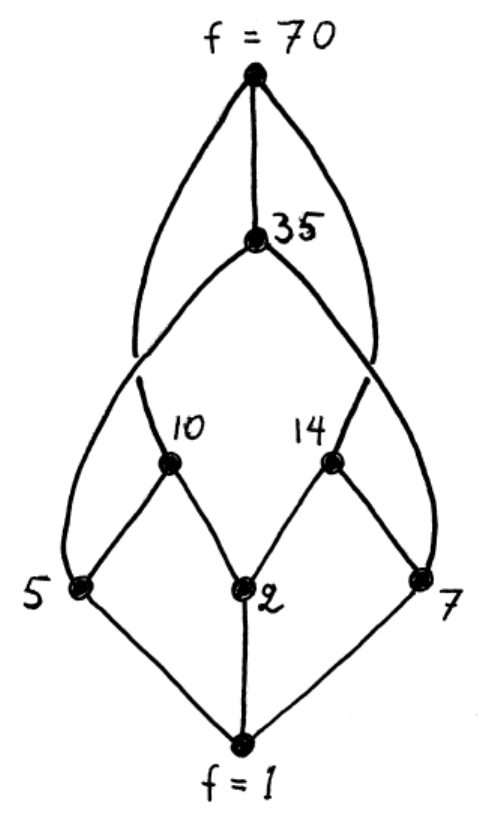} } }
\caption{\small{The lattice of divisors $f$ of $f_0 = 70$ for $d = 199$, with 
$K=\Q(\sqrt{2})$. 
For $f = 1$ the class number is just $h_K$ which is 1. 
For $f = 70$ it is~$12$, and 
the excess $e = 12$ too. 
The SICs corresponding to $f = 1$ and~5 
have anti-unitary symmetry and 
have been constructed from Stark units in exact form. The remaining 
examples are predictions by Kopp and Lagarias, most of which have 
been verified in unpublished work by Grassl.}} 
\label{fig:gitter199}
\end{figure}

Kopp and Lagarias go on to claim that for each divisor $f$ of $f_0$ there 
exists a SIC that can be constructed using their new ray class field $K_L^{d;f}$ 
\cite{Kopp2}. For $d \leq 90$ the resulting `spectrum' of geometrically 
inequivalent SICs agrees with the numerical findings of Scott \cite{Scott2}, 
which are believed to be complete for these dimensions. Here it should be 
recalled that SICs are collected first into anti-unitarily equivalent 
multiplets based on their symmetries, and then these geometric multiplets 
are collected into Galois multiplets where the number of geometric 
multiplets in a given Galois multiplet is equal to the class number $h_K$ of 
the quadratic field~$K$ if the SICs are minimal, and equal to the degree of 
a ring class field in the general case. The conjecture evidently has 
predictive content. We give the example $d = 199$ in Figure 
\ref{fig:gitter199}. Equation (\ref{Eq1}) yields $f_0 = 70$, so there is 
a lattice of no less than eight divisors. The SICs with $f = 1$ 
and $f = 5$ are known. 
The remaining cases have not yet been publicly reported. 

The ray class fields with modulus $d$ and conductor $f$ are subfields of 
the ray class field with modulus $fd$ and conductor 1, and often 
proper subfields. Nevertheless our thesis is that the SIC overlaps are 
always given by products of Stark units in the ray class field $K^{fd \infty_1}$ 
and in some of its subfields. Before we present the evidence we will prove 
a corollary to a theorem in ref. \cite{Kopp2}. First we define the totally positive unit 

\begin{equation} u = \frac{d-1 + \sqrt{(d+1)(d-3)}}{2} \ . 
\label{unitu} \end{equation}

\noindent Whether or not this is a fundamental unit in the quadratic field~$K$ 
depends on the position of $d$ in the AFMY tower of dimensions 
connected to~$K$ (and on whether or not  
$d-3$ is a square) \cite{AFMY}. If $d = d_1$ it equals $u_D$, the first 
totally positive power of a fundamental unit. 

Next, let $\Phi(f)$ be Euler's totient function and 

\begin{equation} \Phi_K(f) = f^2\prod_{p \lvert f}\left(1 - \frac{1}{p}\right) 
\left( 1 - \mbox{$\underline{\Delta_0} \choose p$}
\frac{1}{p}\right) \end{equation}

\noindent its generalisation to quadratic fields. Here 
${\underline{\Delta_0} \choose p}$ is the Kronecker symbol, equal to the Legendre 
symbol if $p$ is odd, while if $p=2$ it equals 0 if $\Delta_0$ is even, 1 if 
$\Delta_0 = \pm 1$ mod 8, and $-1$ if $\Delta_0 = \pm 3$ mod 8. It tells us 
whether the prime $p$ ramifies, splits, or remains inert in~$K$. 

There will be a complication if $d$ and $f$ have a common factor. By inspection 
of the key equation (\ref{Eq1}) we see that $\gcd(d,f) \in \{1,3\}$, and if 
$\gcd(d,f) = 3$ then $d \equiv 3$ mod 9. 
We will largely ignore this special case until we come to Section \ref{sec:sec7}, 
but it will be included in the statement of the proposition. 

We can now state our observation:

\

\begin{proposition}
{\sl If $d$ and $f$ are relatively prime 
the excess as defined above is} 
\begin{equation} \mbox{excess} = \left\{ \begin{array}{lc} \Phi (f) & 
\mbox{if} \ u^3 \equiv 1 \bmod fd \\ \\ \Phi(f)/2 & \mbox{otherwise} \ . 
\end{array} \right. \end{equation} 
\noindent {\sl If $d$ and $f$ are both divisible by 3 then the excess is $3/2$ times 
this in each case.} 
\end{proposition}

\

\paragraph{Proof:}
We rely on a result by Kopp and Lagarias, namely 
that (see Theorem 6.5 in ref. \cite{Kopp2})  

\begin{equation} \mbox{degree}(K_L^{d, \Sigma ; f}) = \left\{ \begin{array}{lcl} 
\frac{2^{ \lvert \Sigma  \lvert }h\Phi_K(f)}{\Phi (f)}\times 
\mbox{degree}(K^{d}) & \mbox{if} & (d,f) = 1 \\ \\ 
\frac{2}{3}\frac{2^{ \lvert  \Sigma  \lvert }h\Phi_K(f)}{\Phi (f)}\times 
\mbox{degree}(K^{d}) & \mbox{if} & (d,f) = 3 \end{array} \right. \end{equation}

\noindent Here $ \lvert \Sigma \lvert $ is the number of infinite places, and $h_K = \lvert\cK\lvert$ is 
the class number of~$K$. For the ray class field with finite 
modulus $d$ ramified at one infinite place we know (see ref. \cite[Ch VI \S1 eq~(15)]{Lang}), that 

\begin{equation} \mbox{degree}(K^{d\infty_1}) = \frac{2h\Phi_K(d)}
{[U_K:\U{d}]} \ , \label{Lang} \end{equation}

\noindent where $[U_K:\U{d}]$ is the index of the subgroup of the unit group $U_K$ 
consisting of units that equal 1 mod $d$. Putting things together, and 
using the multiplicative property of $\Phi_K$, we find

\begin{equation} {\rm excess} = \Phi(f) \frac{[U_K:\U{d}]}{[U_K:\U{fd}]} \ , 
\label{index} \end{equation}

\noindent unless $d$ and $f$ have a common factor in which case we obtain

\begin{equation} {\rm excess} = \frac{3}{2} 
\Phi(f) \frac{[U_K:\U{d}]}{[U_K:\U{fd}]} \ , \label{index3} \end{equation}

\noindent We need to calculate the orders of $U_K/\U{d}$ and $U_K/\U{fd}$. 
Because the torsion group $\pm 1$ is included in $U_K$ this is equal 
to twice the order of a positive fundamental unit. For the unit defined in eq. 
(\ref{unitu}) we calculate 

\begin{equation} \begin{array}{lc} 
u^3 =  1 + \frac{d^2(d-3) + d(d-2)\sqrt{(d+1)(d-3)}}{2} \\ \\ 
u^6 =  1 + \frac{d^2(d+1)(d-3)(d-2)^2 + d(d-1)(d-2)(d^2-2d -2) 
\sqrt{(d+1)(d-3)}}{2} \ . \end{array} \end{equation}

\noindent The unit $u$ may not be a fundamental unit, but this complication 
cancels between numerator and denominator in equation (\ref{index}). By 
inspection we see that $u^3 = 1$ mod $d$ and $u^6 = 1$ mod $fd$. This suffices 
to prove the theorem as stated.  \qed 

\

Our proposition provides a supplement to the more abstract description in ref. 
\cite{Kopp2}. To compute the excess is now a question of checking whether 
$u^3 \equiv 1$ mod $fd$. We find that the excess equals 1 for $f = 2$ whenever it occurs, 
namely for $d \equiv 3$ mod 4. It may be equal to 1 also if $\Phi (f) = 2$, namely if 
$f = 3$, 4, or 6. Table \ref{tab:KoppL} gives the excess in a few 
interesting cases. 

{\small 
\begin{table}[h]
\caption{{\small The excess for some values of $f$.}}
 \smallskip \smallskip
\hskip 1.6cm
{\renewcommand{\arraystretch}{1.2}
\begin{tabular}
{|c|c|c||c|c| c|}\hline 
$f$ & Dimensions  & excess & $f$ & Dimensions & excess \\
\hline \hline 
2 & $d = 7 + 4k$  & 1 & 8 & 63 + 64k & 2 \\
3 & $d = 8 + 9k$ & 1 & & 3 + 64k & 4 \\ 
\ & $d = 3 + 9k$ & 3 & & 19 + 64k & 2 \\
4 & $d = 15 + 16k$  & 1 & & 47 + 64k & 2 \\
 & $d = 3+16k$ &  2 & 15 & 224 + 225k & 4 \\ 
5 & $d = 24 + 25k$ & 2 & & 3 + 225k & 12 \\ 
 & $d = 3 + 25k$ & 4 & & 53 + 225k & 4 \\ 
6 & $d = 35 + 36k$ & 1 & & 174 + 225k & 6 \\ 
 & $d = 3 + 36k$ & 3 & & & \\ 
\hline 
\end{tabular}
}
\label{tab:KoppL}
\end{table}
}

There are two further points to make. First, for minimal SICs 
we observed that the Galois group Gal$(H/K)$, where $H$ is the Hilbert 
class field, gives rise to multiplets of geometrically inequivalent SICs. 
For non-minimal SICs the role of the Hilbert class field is taken over by 
a ring class field defined for the particular subring of integers that 
is defined by the conductor $f$. 

Then we will be interested to know for what values of $d$ and $f$  
the SIC exhibits anti-unitary symmetry. This happens only if the 
underlying quadratic field~$K$ has a fundamental unit of negative norm, which 
means that the degree of the ray 
class field with modulus $d$ goes down by a factor of 2. The SIC 
responds to this by having an anti-unitary symmetry, so that the number of 
distinct overlaps decreases by a factor of 2. In general, if $d$ 
is of the form $n^2+3$ then 
a unit of negative norm is 

\begin{equation} u_0 = \frac{n + \sqrt{n^2+4}}{2} = 
\frac{n + f_u\sqrt{\Delta_0}}{2} \ . \end{equation}

\noindent Squaring it, we obtain the totally positive unit 

\begin{equation} u = u_0^2 = \frac{n^2+2 + n\sqrt{n^2+4}}{2} = 
\frac{d-1 + nf_u\sqrt{\Delta_0}}{2} \ , \end{equation}

\noindent where $f_u$ is an integer. According to Kopp and Lagarias \cite{Kopp2} 
non-minimal SICs are attached to a non-maximal ring of integers, and 
(conjecturally) non-minimal SICs have anti-unitary symmetry if and only if $u_0$ 
belongs to that ring. It follows that SICs associated to a divisor $f$ will 
have anti-unitary symmetry if and only if $f \lvert f_u$, which is a more stringent condition 
than $f \lvert f_0$ \cite{Kopp4}.

\section{Overlaps for non-minimal SICs}\label{sec:sec6}

\noindent We continue to build our phenomenology through the simple expedient of 
calculating examples. For this purpose we need to know what examples we can 
realistically expect to calculate. 
Tables \ref{tab:numerologi} and \ref{tab:numerologi2} in Appendix \ref{sec:A3} list low 
dimensional examples for which non-minimal SICs occur. The degrees of the number 
fields in which we are interested tend to grow (quickly) with $d$ and $f$, so 
this gives some idea of where to find accessible examples. When anti-unitary 
symmetry occurs we can go to high dimensions. 

What we find is that the overlap units for a SIC with a given conductor $f >1$, 
having divisors $f_1, \dots , f_n$ where $f_1 = 1$ and $f_n = f$, are products 
of overlap units for SICs with conductors that divide $f$, times a new unit 
coming from the field with conductor equal to $f$ itself. If the excess $e > 1$ this 
new unit cannot be a Stark unit from $K^{fd \infty_1}$, since this would generate too 
large a field. Let $K_L^{d;f}$ be the Kopp--Lagarias field in our homemade notation. 
The idea is that 

\begin{equation}  \lvert \mbox{Gal}(K^{fd}/K_L^{d;f)} \lvert  = e \ . \end{equation}

\noindent If $e = 1$ the new unit is, as expected, a square root of a Stark unit 
in $K^{fd \infty_1}$ raised to some suitable power. 
But suppose that the excess is $e > 1$. What we can do then 
is to calculate Stark units in $K^{fd \infty_1}$, identify the Galois 
group of the extension $K^{fd \infty_1}/K_L^{d \infty_1; f}$, and act by 
it on one of the Stark units that we calculated. This gives an orbit consisting of $e$ units, 
and we multiply all of them together so that we obtain a unit left invariant by 
this Galois group. In other words, we take a {\it relative 
norm}. We end up with a collection of $e$-plets of Stark units that lie 
in $K_L^{d \infty_1; f}$, and these $e$-plets are the very units that appear in the overlaps. 
We furthermore find that Grassl's rule (\ref{Grasslrule}) correctly determines all the 
exponents of the square rooted Stark units that appear, provided it is applied 
separately to the contributions from each subfield, and provided we remember 
that the relevant Galois group is that which keeps the ring class field fixed. 

The examples support this simple picture quite consistently. A drawback is 
evidently that we have to perform calculations in fields of high 
degree. Also, in practice it may not be so easy to identify the 
Kopp--Lagarias field. 
On the bright side it may happen that even if 
$e > 1$ the excess for a subfield giving some of the baby overlaps may equal 
1. An example is $d = \partial \overline{\partial}= 199$, $f = 5$, $e = 2$, 
where the subfield $K^{5\partial, \infty_1}$ has excess 1. In this case 
an exact SIC fiducial vector with anti-unitary symmetry is easy to compute using the algorithm in ref. 
\cite{ABGHM}. 

As usual dimensions $d = p \equiv 2$ mod 3 show the simplest pattern, but 
unfortunately we do not find any examples with excess $e > 1$ 
among those we were able to compute. For the examples in Table \ref{tab:1och2mod3} 
we performed the calculations exactly using Kopp's method \cite{Kopp1}. 
Using Magma for $d = 23$, $f = 2$, it took 14 hours to 
compute the Stark units in the larger field, 12 days to find a root of 
the minimal polynomial, and 29 minutes to verify that we have a SIC. 
The first case with $e>1$ occurs for $d = 47$, $f = 8$, $e = 2$, and 
calculating the minimal polynomial for the Stark units in this case is 
beyond us. In the examples we did calculate there are non-trivial 
ring class fields, giving rise to two (for $d = 11$, 17) and four (for $d = 23$) 
unitarily inequivalent SICs, respectively, in the non-minimal multiplets. 

{\small 
\begin{table}[h]
\caption{{\small Non-minimal SIC overlaps in prime dimensions. The examples 
are $d= 11$, 23; 31; 19, 67, 5779 for $f = 2$, $d = 17$ for $f = 3$, $d = 31$ 
for $f = 4$, and $d = 199$ for $f = 5$. The seventh and last row was checked for $d = 19$ 
only, and does not apply for $d = 199 = 14^2 + 3$ because $K^{5\cdot 199\infty_1}$ 
has $e = 2$.}}
  \smallskip \smallskip
\hskip 0.5cm
{\renewcommand{\arraystretch}{1.6}
\begin{tabular}
{|c|c|c|c|c|c|}\hline 
Dimension & Type & \#  & $f=1$ & $f = 2$, 3, 5 & $f = 4$ \\
\hline \hline
$p=2$ mod 3 & $D$ & $p^2-1$ & $(\mbox{S}_{p})^\frac{1}{2}$ & $(\mbox{S}_{p}
\mbox{S}_{fp})^\frac{1}{2}$ & -- \\ 
\hline \hline 
$p =1$ mod 3 & $Z$ & $p-1$ & $(\mbox{S}_{\bar{\partial}})^\frac{1}{2}$ & 
$(\mbox{S}_{\bar{\partial}}\mbox{S}_{f\bar{\partial}})^\frac{1}{2}$ & 
$(\mbox{S}_{\bar{\partial}}\mbox{S}_{f\bar{\partial}}
\mbox{S}_{4\bar{\partial}})^\frac{1}{2}$ \\
$p \neq n^2+3$ & $X$ & $p-1$ & $(\mbox{S}_{\partial})^\frac{1}{2}$ & 
$(\mbox{S}_{\partial}\mbox{S}_{f\partial} )^\frac{1}{2}$ & 
$(\mbox{S}_{\partial}\mbox{S}_{f\partial}\mbox{S}_{4\partial} )^\frac{1}{2}$ \\ 
& $D$ & $(p-1)^2$ & $(\mbox{S}_p)^\frac{1}{2}$ & 
$(\mbox{S}_{p}\mbox{S}_{fp})^\frac{1}{2}$ & \\ 
\hline \hline 
$p=1$ mod 3 & $Z$ & $p-1$ & 1 & 1 & -- \\
$p=n^2+3$ & $X$ & $p-1$ & $\mbox{S}_{\partial}$ & 
$\mbox{S}_{\partial}\mbox{S}_{f \partial}$ & -- \\
& $D$ & $(p-1)^2$ & $(\mbox{S}_p)^\frac{1}{2}$ & 
$(\mbox{S}_p\mbox{S}_{fp})^\frac{1}{2}$ & -- \\ 
\hline 
\end{tabular}
}
\label{tab:1och2mod3}
\end{table} 
}

To proceed we again lower our standards. We rest content with a check 
that the numerical overlap units provide roots of the minimal polynomial of 
an appropriate combination of Stark units. If it is too time-consuming to calculate all the 
Stark units we rest content with analysing some of the baby overlaps, 
and leave some entries in the tables blank.

The case $d = p \equiv 1$ mod 3 splits into two subcases depending on whether 
$d-3$ is a square or not. See Table \ref{tab:1och2mod3}. When $d = 31$ we come across $f = 4$. 
This has two divisors, and we have an example of a non-minimal SIC whose overlaps are 
given by products of Stark units from three different fields. The examples for 
$d-3$ being a square have anti-unitary symmetry, of order 18 for $d = 5779$, $f = 2$.  

We need to consider some examples with excess $e > 1$, so that the 
Kopp--Lagarias ray class fields are not the usual ones. Scanning Tables 
\ref{tab:numerologi} and \ref{tab:numerologi2} for accessible examples 
we find $d = 19$ and $d = 28$. The non-minimal SICs with $f > 2$ that we are about to 
consider have symmetry of order 3 only. 

{\small 
\begin{table}[ht]
\caption{{\small Overlaps for $d = n^2 + 3 =p = 1$ mod 3 when 
the SIC does not have anti-unitary symmetry. The example is $d = 19$, 
with $e = 2$ in both cases. When $f = 1$ or 2 the SIC has anti-unitary 
symmetry, and this has consequences for the way $S_\partial$ and 
$S_{2\partial}$ occur in the baby overlaps.}}
  \smallskip \smallskip
\hskip 1.0cm
{\renewcommand{\arraystretch}{1.6}
\begin{tabular}
{|c|c|c|c|}\hline 
Type & \# & $f=4$ & $f=8$ \\
\hline \hline
$Z$ & $p-1$ & $(\mbox{S}^\prime_{4\bar{\partial}}
\mbox{S}^{\prime \prime}_{4\overline{\partial}})^\frac{1}{2}$ & 
$(\mbox{S}^\prime_{4\overline{\partial}}\mbox{S}^{\prime \prime}_{4\overline{\partial}}
\mbox{S}^\prime_{8\overline{\partial}}\mbox{S}^{\prime\prime}_{8\overline{\partial}}
)^\frac{1}{2}$ \\
$X$ & $p-1$ & 
$\mbox{S}_{\partial}\mbox{S}_{2 \partial}
(\mbox{S}^\prime_{4\partial}\mbox{S}^{\prime \prime}_{4\partial})^\frac{1}{2}$ & 
$\mbox{S}_{\partial}\mbox{S}_{2\partial}(\mbox{S}^\prime_{4\partial}
\mbox{S}^{\prime \prime}_{4\partial}
\mbox{S}^\prime_{8\partial}\mbox{S}^{\prime \prime}_{8\partial})^\frac{1}{2}$ \\
$D$ & $(p-1)^2$ 
& $(\mbox{S}_p\mbox{S}_{2p}\mbox{S}^\prime_{4p}\mbox{S}^{\prime \prime}_{4p}
)^\frac{1}{2}$ & $(\mbox{S}_p\mbox{S}_{2p}\mbox{S}^\prime_{4p}
\mbox{S}^{\prime \prime}_{4p}\mbox{S}^\prime_{8p}\mbox{S}^{\prime \prime}_{8p}
)^\frac{1}{2}$\\ 
\hline 
\end{tabular}
}
\label{tab:19faser}
\end{table} 
}

{\small 
\begin{table}[h]
\caption{{\small Overlaps for $d = 28 = 4p = 4\partial \bar{\partial}$, $f = 5$. 
There is no anti-unitary symmetry, and the excess $e = 4$. The table is incomplete. 
For $f=1$ a complete table (with trivial overlaps excluded) is given in Table 
\ref{tab:4pfaser}.}} \smallskip \smallskip
\hskip 1.3cm
{\renewcommand{\arraystretch}{1.0}
\begin{tabular}
{|c|c|c|c|}\hline 
Type & \# & $f = 1$ & $f = 5$ \\
\hline \hline
${\bf 1} \otimes Z$ & $p-1$ & 1 & 
$(\mbox{S}^{\prime}_{5\overline{\partial}}\mbox{S}^{\prime \prime}_{5\overline{\partial}}
\mbox{S}^{\prime \prime \prime}_{5\overline{\partial}}
\mbox{S}^{\prime \prime \prime \prime}_{5\overline{\partial}})^\frac{1}{2}$ \\
$D_{0,2}\otimes Z$ & $3(p-1)$ & $-1$ & 
$(\mbox{S}^{\prime}_{10\overline{\partial}}
\mbox{S}^{\prime \prime}_{10\overline{\partial}}
\mbox{S}^{\prime \prime \prime}_{10\overline{\partial}}
\mbox{S}^{\prime \prime \prime \prime}_{10\overline{\partial}})^\frac{1}{2}$ \\
${\bf 1} \otimes X$ & $p-1$ & $\mbox{S}_{\partial}$ & 
$S_\partial (\mbox{S}^{\prime}_{5\partial}\mbox{S}^{\prime \prime}_{5\partial}
\mbox{S}^{\prime \prime \prime}_{5\partial}
\mbox{S}^{\prime \prime \prime \prime}_{5\partial})^\frac{1}{2}$ \\
$D_{0,2}\otimes X$ & $3(p-1)$ & $\mbox{S}_{2\partial}$ & 
$S_{2\partial} (\mbox{S}^{\prime}_{10\partial}\mbox{S}^{\prime \prime}_{10\partial}
\mbox{S}^{\prime \prime \prime}_{10\partial}
\mbox{S}^{\prime \prime \prime \prime}_{10\partial})^\frac{1}{2}$ \\
${\bf 1}\otimes D$ & $(p-1)^2$ & $(\mbox{S}_{p})^\frac{1}{2}$ & 
$(\mbox{S}_{p}\mbox{S}^{\prime}_{5p}\mbox{S}^{\prime \prime}_{5p}
\mbox{S}^{\prime \prime \prime}_{5p}\mbox{S}^{\prime \prime \prime \prime}_{5p})^\frac{1}{2}$ \\ 
$D_{0,2}\otimes D$ & $3(p-1)^2$ & $(\mbox{S}_{2p})^\frac{1}{2}$ & 
$(\mbox{S}_{2p}\mbox{S}^{\prime}_{10p}\mbox{S}^{\prime \prime}_{10p}
\mbox{S}^{\prime \prime \prime}_{10p}\mbox{S}^{\prime \prime \prime \prime}_{10p})^\frac{1}{2}$\\ 
$D\otimes {\bf 1}$ & 12 & $(\mbox{S}_4)^\frac{1}{2}$ & 
$(\mbox{S}_4\mbox{S}^{\prime}_{20}\mbox{S}^{\prime \prime}_{20}
\mbox{S}^{\prime \prime \prime}_{20}\mbox{S}^{\prime \prime \prime \prime}_{20})^\frac{1}{2}$ \\ 
\hline 
\end{tabular}
}
\label{tab:28faser}
\end{table}
}

For $d = 19$ we have $f_0 = 8$, so there are four possible values of $f$. See 
Table \ref{tab:1och2mod3}, and Table \ref{tab:19faser} for examples with $e= 2$. When 
$f = 8$ Stark units from four distinct fields appear. An example with a more interesting 
divisor lattice for $f_0$ will be found in Table \ref{tab:35faser} below. 

An example with $e = 4$ is $d = 28$, 
$f = 5$. We bring up this case also because repeated Stark units 
occur---the field $K^{5\partial, \infty_1}$ has degree 16 over $K$, but there 
are only eight distinct Stark units and they do not generate the field. This is 
the largest ray class field in which we have encountered this phenomenon: 
one which is (conjecturally) explained by extra zeroes 
in the L-functions and which will be the subject of a future work. One 
finds when working out the relevant third row in Table 
\ref{tab:28faser} that although the 4-plet of Stark units seemingly has too 
low a degree over $K$, its degree over the ring class field---which in this case equals 
${\mathbb{Q}}(\sqrt{5})$---is the same as that of the 4-plet occurring in the first 
row. And this is what counts when applying the rule (\ref{Grasslrule}), which correctly 
predicts the exponents in all the rows. Anyway the degree over $K$ rises 
to the expected value when the 4-plet is multiplied with a Stark unit squared 
coming from $K^{\partial  \infty_1}$. 

\section{Dimensions of the form $d = 3(3n+1)$}\label{sec:sec7}

\noindent It was noted in Section \ref{sec:sec3} that dimensions divisible by 3 
but not by 9 are special because there is more than one conjugacy class of 
possible order three symmetries. This happens because the order three symmetry can act 
like the identity in the dimension 3 factor of Hilbert space---this is allowed because the 
symplectic identity matrix has trace $-1$ counted mod 3.  In Section \ref{sec:sec5} we found that 
dimensions $d = 3(3n+1)$ are special because it can happen that $f_0$ admits a 
divisor $f$ that also divides $d$. 
The upshot of this is that dimensions of the form $d = 3(3n+1)$ admit SICs of 
a different symmetry type known as {\it type $F_a$} \cite{Scott1}. The ones with an 
`ordinary' symmetry are referred to as type $F_z$. One finds in these dimensions that 
the minimal SIC is of type $F_a$, but this is always accompanied 
by non-minimal SICs of type $F_z$. 

The discussion of the centraliser $M(S)$ in Section \ref{sec:sec3} now needs 
amendment. Clearly any $GL(2,{\mathbb{Z}}/3{\mathbb{Z}})$ matrix commutes with the identity 
matrix. To realize the isomorphism with the abelian Galois group we need 
to choose a maximal abelian subgroup of the centraliser in the dimension 3 factor. 
Up to conjugation there are three possibilities \cite{ACFW}. The generators of the three 
abelian groups are respectively 

\begin{eqnarray} 
a_4: \hspace{3mm} 
\left\langle \left( \small{ \begin{array}{cc} -1 & 0 \\ 0 & - 1 \end{array}} \right) \ , 
\left( \small{ \begin{array}{cc} 1 & 0 \\ 0 & - 1 \end{array}} \right) \right\rangle 
\hspace{12mm} \nonumber \\ \\ a_6: \hspace{3mm} 
\left\langle \left( \small{ \begin{array}{cc} -1 & 0 \\ 1 & - 1 \end{array}} \right)\right\rangle 
\ , \hspace{6mm} a_8: \hspace{3mm} 
\left\langle \left( \small{ \begin{array}{cc} 1 & -1 \\ 1 & 1 \end{array}} \right) \right\rangle 
\ . \nonumber \end{eqnarray}

\noindent When acting on the displacement operators in the dimension 3 
factor these groups divide the 8 non-trivial ones into $2 + 2 + 4$ for subtype ${a_4}$ 
and into $2 + 6$ for ${a_6}$ (as happens also for $F_z$ SICs), while the $a_8$ 
group acts transitively on all eight of them. 

{\small 
\begin{table}[h]
\caption{{\small Some minimal SIC overlaps for SICs of type $F_a$. The examples 
are $d = 30$ for $a_6$, $d = 66$ for $a_4$, and $d = 21$ for $a_8$. For $F_a$ SICs 
with anti-unitary symmetry, see Table \ref{tab:3pbaby}.}}
  \smallskip \smallskip
\hskip 1.9cm
{\renewcommand{\arraystretch}{1.6}
\begin{tabular}
{|c|c|c|c|c|c|c|}\hline 
$D$ mod 3 & Type & \#  & $ \lvert \mbox{Gal} \lvert $ & $ \lvert S \lvert $ & Overlap \\
\hline \hline 
0, $a_6$ & $Z\otimes {\bf 1}$ & 2 & 2 & 3 & $(\mbox{S}_\partial)^\frac{3}{2}$ \\ 
 & $D\otimes {\bf 1}$ & 6 & 6 & 3 & $(\mbox{S}_3)^\frac{3}{2}$ \\ 
\hline \hline 
1, $a_4$ & $Z\otimes {\bf 1}$ & $2$ & 2 & $3$ & $(\mbox{S}_{\overline{\partial}_3})^\frac{3}{2}$ \\ 
 & $X\otimes {\bf 1}$ & $2$ & 2 & $3$ & $(\mbox{S}_{\partial_3})^ \frac{3}{2}$ \\ 
 & $D\otimes {\bf 1}$ & $4$ & 4 & $3$ & $(\mbox{S}_{3})^\frac{3}{2}$ \\ 
\hline \hline 
2, $a_8$ & $D\otimes {\bf 1}$ & $8$ & 8 & 3 & $(\mbox{S}_3)^\frac{3}{2}$ \\ 
\hline 
\end{tabular}
}
\label{tab:Fa}
\end{table} 
}

This will have consequences for the Galois orbits of the baby overlaps. In 
particular, once the Hilbert space has been expressed as ${\mathbb{C}}^3\otimes 
{\mathbb{C}}^{3k+1}$ it will have consequences for overlaps of the form 
$\langle \Psi_0  \lvert D_{i,j}\otimes {\bf 1} \lvert \Psi_0 \rangle$. They will be formed 
from Stark units in $K^{3, \infty_1}$ and its subfields. The divisor 
lattices shown in Figure \ref{fig:gitterp1} are relevant here. 
If 3 is inert---that is, if it remains a prime over the 
quadratic field ${\mathbb{Q}}(\sqrt{D})$---then the 
Galois group will act transitively on these overlaps and the subtype must be 
$a_8$. If 3 splits there will be three orbits of non-trivial overlaps and the 
subtype must be $a_4$. And if 3 ramifies there will be two orbits, and the subtype 
must be $a_6$ \cite{Marcusray}. This behaviour is determined by the value of $D$, 
according to 
\begin{equation} \begin{array}{lll} D = 0 \ \mbox{mod} \ 3 & \Rightarrow 
\ \mbox{3 ramifies} & \Rightarrow \ a_6 \\ \\ 
D = 1 \ \mbox{mod} \ 3 & \Rightarrow 
\ \mbox{3 splits} & \Rightarrow \ a_4 \\ \\ 
D = 2 \ \mbox{mod} \ 2 & \Rightarrow 
\ \mbox{3 is inert} & \Rightarrow \ a_8 \ . \end{array} \end{equation}

{\small 
\begin{table}[h]
\caption{{\small Non-trivial minimal $F_a$ SIC overlaps for $d = n^2 + 3 =3p$. 
The primes split according to $3 = \partial_3\overline{\partial}_3$ and 
$p=\partial_p\bar{\partial}_p$. The examples are $d = 39$, 327, 1299, with $\ell = 3$ 
for the third example. Rows where the degree grows like $(p-1)^2$ were checked 
only for $d = 39$.}}
  \smallskip \smallskip \smallskip
\hskip 1.6 cm
{\renewcommand{\arraystretch}{1.0}
\begin{tabular}
{|c|c|c|c|c|}\hline
Type & \# & $ \lvert \mbox{Gal} \lvert $ & $ \lvert S \lvert $ & Overlap \\ 
\hline \hline
${\bf 1}\otimes X$ & $p-1$ & $(p-1)/3\ell$ & $6\ell$ & $\mbox{S}_{\partial_{p}}$ \\ 
${\bf 1}\otimes D$ & $(p-1)^2$ & $(p-1)^2/6\ell$ & $6\ell$ 
& $(\mbox{S}_{p})^\frac{1}{2}$ \\ 
$Z\otimes X$ & $2(p-1)$  & $(p-1)/3\ell$ & $6\ell$ 
& $(\mbox{S}_{\overline{\partial}_3\partial_{p}})^\frac{1}{2}$ \\ 
$Z\otimes D$ & $2(p-1)^2$ & $(p-1)^2/3\ell$ & $6\ell$ & 
$(\mbox{S}_{\overline{\partial}_3p})^\frac{1}{2}$ \\ 
$X\otimes {\bf 1}$ & 2 & $2$ & $6\ell$ & $(\mbox{S}_{\partial_3})^{3\ell}$ \\ 
$X\otimes Z$ & $2(p-1)$ & $(p-1)/3\ell$ & $6\ell$ 
& $(\mbox{S}_{\partial_3\overline{\partial}_{p}})^\frac{1}{2}$ \\ 
$X\otimes X$ & $2(p-1)$ & $2(p-1)/3\ell$ & $6\ell$ & $\mbox{S}_{\partial_3\partial_{p}}$ \\ 
$X\otimes D$ & $2(p-1)^2$ & $(p-1)^2/3\ell$ & $6\ell$ 
& $(\mbox{S}_{\partial_3p})^\frac{1}{2}$ \\ 
$D\otimes {\bf 1}$ & 4 & $2$ & $6\ell$ & $(\mbox{S}_3)^\frac{3\ell}{2}$ \\ 
$D\otimes Z$ & $4(p-1)$ & $2(p-1)/3\ell$ & $6\ell$ & 
$(\mbox{S}_{3\overline{\partial}_{p}})^\frac{1}{2}$ \\ 
$D\otimes X$ & $4(p-1)$ & $2(p-1)/3\ell$ & $6\ell$ & 
$(\mbox{S}_{3\partial_{p}})^\frac{1}{2}$ \\ 
$D\otimes D$ & $4(p-1)^2$ & $2(p-1)^2/3\ell$ & $6\ell$ & $(\mbox{S}_d)^\frac{1}{2}$ \\
\hline 
\end{tabular}
}
\label{tab:3pbaby}
\end{table} 
} 

{\small
\begin{table}[h]
\caption{{\small Overlaps for non-minimal $F_a$ and $F_z$ SICs when $d = 39$, $f = 1$ is $F_a$ of subtype 
$a_4$, $f = 2$ is $F_a$, and $f = 3$ is $F_z$. The Hilbert space is ${\mathbb{C}}^{39} = 
{\mathbb{C}}^3\otimes {\mathbb{C}}^{13}$. The dimension splits into four 
non-principal ideals, $d = 3\cdot 13 = \partial_3\overline{\partial}_3\partial_{13}
\overline{\partial}_{13}$. There are identities of the form S$_3 = \mbox{S}_{2\partial_3}$ 
and $\mbox{S}_9^\prime \mbox{S}_9^{\prime \prime}\mbox{S}_9^{\prime \prime \prime} 
= \mbox{S}_{3\partial_3} = \mbox{S}_{3\overline{\partial}_3}$. 
The table is incomplete.}}
 \smallskip \smallskip
\hskip 0.0cm
{\renewcommand{\arraystretch}{1.0}
\begin{tabular}
{|c|c|c|c||c|c|c|}\hline 
Type & \# & $F_a$, $f = 1$ & $F_a$, $f = 2$ & Type & \# & $F_z$, $f = 3$ \\
\hline \hline
${\bf 1} \otimes Z$ & 12 & $1$ & $(\mbox{S}_{2\overline{\partial}_{13}})^\frac{1}{2}$ & 
${\bf 1}\otimes Z$ & 12 & 
$(\mbox{S}_{3\overline{\partial}_{13}}^\prime \mbox{S}_{3\overline{\partial}_{13}}^{\prime \prime})^\frac{1}{2}$ \\ 
${\bf 1} \otimes X$ & 12 & $\mbox{S}_{\partial_{13}}$ & 
$\mbox{S}_{\partial_{13}}(\mbox{S}_{2\partial_{13}})^\frac{1}{2}$ & ${\bf 1}\otimes X$ & 12 & 
$\mbox{S}_{\partial_{13}}
(\mbox{S}_{3\partial_{13}}^\prime \mbox{S}_{3\partial_{13}}^{\prime \prime})^\frac{1}{2}$ \\ 
${\bf 1}\otimes D$ & 144 & $(\mbox{S}_{13})^\frac{1}{2}$ & 
$(\mbox{S}_{13}\mbox{S}_{2\cdot 13})^\frac{1}{2}$ & ${\bf 1}\otimes D$ & 144 & 
$(\mbox{S}_{13}\mbox{S}_{3\cdot 13}^\prime 
\mbox{S}_{3\cdot 13}^{\prime \prime})^\frac{1}{2}$ \\ 
$Z\otimes {\bf 1}$ & 2 & $1$ & $(\mbox{S}_{2\overline{\partial}_3})^\frac{3}{2}$ & & & \\ 
$X\otimes {\bf 1}$ & 2 & $(\mbox{S}_{\partial_{3}})^3$ & 
$(\mbox{S}_{2\partial_3})^\frac{3}{2}$ & $X\otimes {\bf 1}$ & 2 & 
$(\mbox{S}_{\partial_3})^3(\mbox{S}_{3\partial_3})^\frac{3}{2}$ \\ 
$D\otimes {\bf 1}$ & 4 & $(\mbox{S}_3)^\frac{3}{2}$ & 
$(\mbox{S}_3\mbox{S}_{2\cdot 3})^\frac{3}{2}$ & $D\otimes {\bf 1}$ & 6 & 
$(\mbox{S}_9^\prime \mbox{S}_9^{\prime \prime}\mbox{S}_9^{\prime \prime \prime})^\frac{1}{2}$ \\ 
\hline 
\end{tabular}
}
\label{tab:39faser}
\end{table}
}

\noindent This is illustrated by Table \ref{tab:Fa}. Note that $d = n^2 +3 = 3$ 
mod 9 happens whenever $3 \lvert n$, which means that quite a few such $F_a$ SICs have 
been constructed with exact arithmetic. The subtype is then necessarily $a_4$. 
See Table \ref{tab:3pbaby}. 

Dimension $d = 12 = 3\cdot 4$ is a very special case. It is not just that `short' 
orbits occur, so that a factor of 3 occurs in the exponents of the Stark units. 
It is that the Stark units S$_{12}$ and S$_{4\overline{\partial}}$ do not generate their 
fields. Indeed their degree is only one half of that of  
$K^{12,\infty_1}$ and $K^{4\bar{\partial} , \infty_1}$ respectively. 
However, taken jointly 
they do generate the full field. Grassl's rule for the exponents holds provided 
the Galois group is that of the ray class field, 
not that of the subfield generated by the Stark units. Another special feature 
of $d = 12$ is the large number of identities between Stark units in the various 
subfields. 

Non-minimal SICs of $F_a$ type appear (or are believed to appear) when 
$f_0/f = 0$ mod 3 \cite{Kopp4}. 
When $f$ and $d$ share a common factor of 3 there is a complication. 
By considering the degrees of the relevant fields one 
is led to expect that all non-minimal $F_a$ SICs with $(d,f) = 3$ are of 
subtype $a_6$. Were this not the case Grassl's rule (\ref{Grasslrule}) 
would predict a fractional exponent. An example of such a non-minimal $F_a$  
SIC occurs when $d = 84$, and it is indeed of subtype $a_6$.  

Whenever an $F_a$ SIC occurs a non-minimal $F_z$ SIC occurs as well. We then 
face the question of how the overlaps in the $F_z$ SICs are related to those 
of the minimal $F_a$ SIC. Table \ref{tab:39faser} gives an account 
of this for $d = 39$. There are no surprises there, but note incidentally 
that for $f = 3$ the excess $e = 3$ for the fields $K^{3d \infty_1}$ and 
$K^{3\cdot 3,\infty_1}$, but $e = 2$ for $K^{3\cdot 13,\infty_1}$. For 
the degree of the Kopp--Lagarias fields $K_L^{m,\infty_1;3}$ it matters 
whether the modulus $m$ is divisible by 3 or not, so this is 
what one would expect given the results in Section \ref{sec:sec5}. But we have no 
proof since the theorem on which Section \ref{sec:sec5} relies \cite{Kopp2} 
covers only the case $m = d$ explicitly. 

\section{SIC alignment}\label{sec:sec8}

\noindent We now consider dimensions of the form $d = n(n-2)$, or equivalently
dimensions such that $d+1$ is a square. This happens if and only if the position $\ell$ 
of the dimension in its AFMY tower is a multiple of 2. If $d_\ell = n$ then $d_{2\ell} = n(n-2)$. 
In these dimensions the phenomenon of SIC alignment enters \cite{Irina}. Actually, 
whenever the base dimension $d_1 \lvert d_\ell$ in an AFMY tower, 
square rooted Stark units from $K^{d_1 \infty_1}$ will appear raised to the power 
$\ell$ among the baby overlaps in dimension $d_\ell$, simply because of Grassl's 
rule (\ref{Grasslrule}) and the fact that the symmetry goes up by such a factor. We 
saw examples ($d = 124$, 844, and 1299) in Tables \ref{tab:4pfaser} 
and \ref{tab:3pbaby}. 

But in the case we 
now have in mind, more is true. It is known that from a SIC in dimension $n$ 
one can construct a continuous family of configurations of vectors in 
dimension $d = n(n-2)$ that share some of the properties of a SIC. If $n$ is 
odd any vector $ \lvert \Psi \rangle$ in this family provably obeys \cite{Basudha} 

\begin{equation} \sqrt{d+1} \langle \Psi \lvert {\bf 1}\otimes D \lvert  \Psi \rangle = - 
(\mbox{overlap unit from dimension $n$})^2 \ , \label{kvadrater}  \end{equation}

\begin{equation} \sqrt{d+1}\langle \Psi  \lvert D\otimes {\bf 1} \lvert \Psi \rangle = 1 \  ,  
\label{ettor} \end{equation}

\noindent where we assumed that the Hilbert space is ${\bf C}^{n-2}\otimes {\bf C}^n$. 
If $n$ is even something similar but a bit more complicated is true \cite{Ole, Danylo}.  It is 
expected that one can always obtain 
a $d$-dimensional SIC by specializing the parameters in this 
construction, and this SIC is then said to be {\it aligned} to the lower 
dimensional SIC. Hence there is a semi-constructive point 
of view on these aligned SICs.

\begin{figure}
        \centerline{ \hbox{
		\includegraphics[width=45mm]{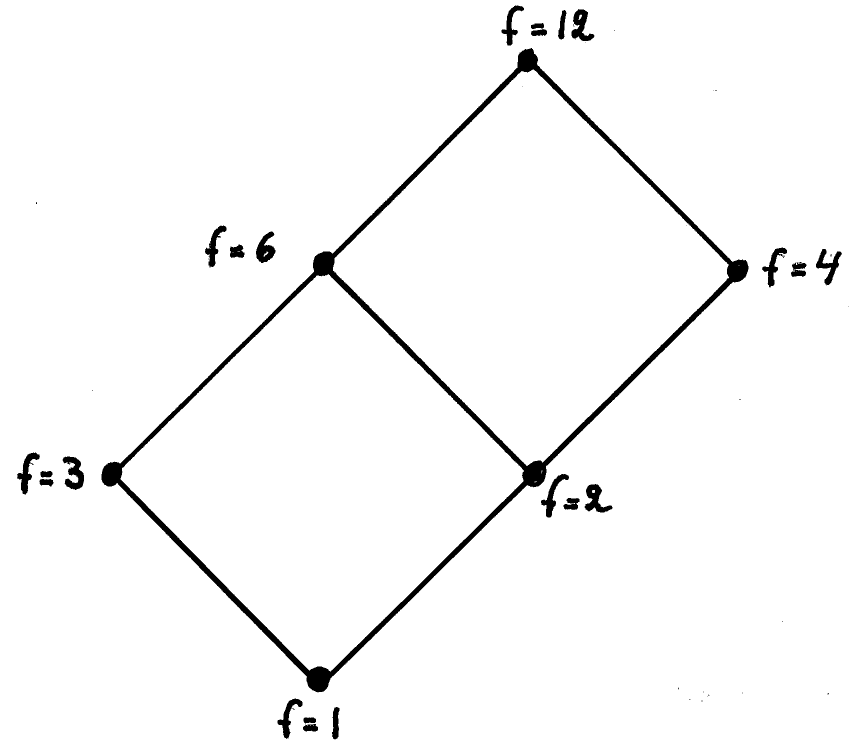} } }
        \caption{\small The inclusion lattice for SIC fields when $d = 35$. The 
				degree of the ring class field is $h = 2$ for $f = 6$, $h = 4$ for 
				$f = 12$, and the excess $e = 2$ for $f = 4$, 12.}
        \label{fig:d35gitter}
\end{figure} 

{\small 
\begin{table}[hb]
\caption{{\small Overlaps for $d = 35$. 
The Hilbert space is ${\mathbb{C}}^5\otimes {\mathbb{C}}^7$. The $f = 1$, 2 SICs 
are aligned with SICs in $d = 7$, and the $f = 1$ SIC inherits anti-unitary 
symmetry from the lower dimension. There are identities of the form 
$\mbox{S}_{4\partial}^\prime \mbox{S}_{4\partial}^{\prime \prime} = 
\mbox{S}_{3\partial} = \mbox{S}_{2\partial} = \mbox{S}_{\partial}^2$, 
$\mbox{S}_{6\overline{\partial}}\prime = \mbox{S}_{6\overline{\partial}}S_{2\overline{\partial}}^2$, 
$\mbox{S}_{4\overline{\partial}}^\prime = \mbox{S}_{4\overline{\partial}}S_{2\overline{\partial}}^2$. 
The Stark units in $K^{4\bar{\partial},\infty_1}$ and  $K^{6\bar{\partial},\infty_1}$ do not 
generate their fields. For $f = 4$, 12, the excess $e = 2$ and we give incomplete results only.}}
  \smallskip \smallskip
{\renewcommand{\arraystretch}{1.0}
\begin{tabular}
{|c|c|c|c|c|}\hline 
Type & $f = 1$ & $f = 2$ & $f = 3$ & $f = 6$ \\
\hline \hline
${\bf 1} \otimes Z$ & $-1$ & $- \mbox{S}_{2\overline{\partial}}$ & 
$(\mbox{S}_{3\overline{\partial}})^\frac{1}{2}$ & $\mbox{S}_{2\overline{\partial}}(\mbox{S}_{3\overline{\partial}}
\mbox{S}_{6\overline{\partial}})^\frac{1}{2}$ \\
${\bf 1} \otimes X$ & $- \mbox{S}_{\partial}^2$ & 
$- \mbox{S}_{\partial}^2\mbox{S}_{2\partial}$ & $\mbox{S}_{\partial}^2(\mbox{S}_{3\partial})^\frac{1}{2}$ & 
$\mbox{S}_{\partial}^2\mbox{S}_{2\partial}(\mbox{S}_{3\partial}\mbox{S}_{6\partial})^\frac{1}{2}$\\
${\bf 1} \otimes D$ & $-\mbox{S}_7$ & 
$-\mbox{S}_7\mbox{S}_{14}$ & $\mbox{S}_7(\mbox{S}_{21})^{\frac{1}{2}}$ 
& $\mbox{S}_7\mbox{S}_{14}(\mbox{S}_{21}\mbox{S}_{42})^{\frac{1}{2}}$ \\ 
$D \otimes {\bf 1}$ & 1 & $1$ & $(\mbox{S}_{15})^\frac{1}{2}$ & 
$(\mbox{S}_{15}\mbox{S}_{30})^\frac{1}{2}$ \\
$D \otimes Z$ & $(\mbox{S}_{5\overline{\partial}})^{\frac{1}{2}}$ & 
$(\mbox{S}_{5\overline{\partial}}\mbox{S}_{10\overline{\partial}})^{\frac{1}{2}}$ & 
$(\mbox{S}_{5\overline{\partial}}\mbox{S}_{15\overline{\partial}})^\frac{1}{2}$ & 
$(\mbox{S}_{5\overline{\partial}}\mbox{S}_{10\overline{\partial}}
\mbox{S}_{15\overline{\partial}}\mbox{S}_{30\overline{\partial}})^\frac{1}{2}$ \\ 
$D \otimes X$ & $(\mbox{S}_{5\partial})^{\frac{1}{2}}$ & 
$(\mbox{S}_{5\partial}\mbox{S}_{10\partial})^{\frac{1}{2}}$ & 
$(\mbox{S}_{5\partial}\mbox{S}_{15\bar{\partial}})^\frac{1}{2}$ & 
$(\mbox{S}_{5\partial}\mbox{S}_{10\partial}
\mbox{S}_{15\partial}\mbox{S}_{30\partial})^\frac{1}{2}$ \\
$D \otimes D$ & $(\mbox{S}_{35})^\frac{1}{2}$ & 
$(\mbox{S}_{35}\mbox{S}_{70})^\frac{1}{2}$ & 
$(\mbox{S}_{35}\mbox{S}_{105})^\frac{1}{2}$ & 
$(\mbox{S}_{35}\mbox{S}_{70}\mbox{S}_{105}\mbox{S}_{210})^\frac{1}{2}$ \\ 
\hline
\end{tabular}
\smallskip
\begin{tabular}
{|c|c|c|}\hline 
& $f=4$ & $f = 12$ \\ 
\hline 
${\bf 1}\otimes Z$ \ & 
$\mbox{S}_{2\overline{\partial}}\mbox{S}_{4\overline{\partial}}$ & 
$\mbox{S}_{2\overline{\partial}}\mbox{S}_{4\overline{\partial}}(\mbox{S}_{3\bar{\partial}}
\mbox{S}_{6\overline{\partial}}
\mbox{S}_{12\overline{\partial}}^\prime 
\mbox{S}_{12\overline{\partial}}^{\prime \prime})^\frac{1}{2}$ \\ 
${\bf 1}\otimes X$ & \ 
$\mbox{S}^2_\partial \mbox{S}_{2\partial} 
(\mbox{S}_{4\partial}^\prime \mbox{S}_{4\partial}^{\prime \prime})^\frac{1}{2}$ \ & 
\ $\mbox{S}_{\partial}^2\mbox{S}_{2\partial}
(\mbox{S}_{4\partial}^\prime \mbox{S}_{4\partial}^{\prime \prime}
\mbox{S}_{3\partial}\mbox{S}_{6\partial}
\mbox{S}_{12\partial}^\prime \mbox{S}_{12\partial}^{\prime \prime})^\frac{1}{2}$ \ \\
${\bf 1}\otimes D$ & 
$\mbox{S}_7\mbox{S}_{14} (\mbox{S}_{28}^\prime \mbox{S}_{28}^{\prime \prime})^\frac{1}{2}$ 
& $\mbox{S}_{7}\mbox{S}_{14}(\mbox{S}_{21}\mbox{S}_{28}^\prime \mbox{S}_{28}^{\prime \prime} 
\mbox{S}_{42}\mbox{S}_{84}^\prime \mbox{S}_{84}^{\prime \prime})^\frac{1}{2}$ \\
$D\otimes {\bf 1}$ & 
$(\mbox{S}_{20}^\prime \mbox{S}_{20}^{\prime \prime})^\frac{1}{2}$ & 
$(\mbox{S}_{15}\mbox{S}_{30}\mbox{S}_{20}^\prime \mbox{S}_{20}^{\prime \prime}
\mbox{S}_{60}^\prime \mbox{S}_{60}^{\prime \prime})^\frac{1}{2}$ \\
\hline 
\end{tabular}
}
\label{tab:35faser}
\end{table} 
}

Equation (\ref{ettor}) will be our main concern in section 9. It is consistent 
with our overall picture only if Stark's construction \cite{Stark1} gives 
trivial units $+1$ for the ray class field $K^{d-2,\infty_1}$. This will be so if 

\begin{equation} K^{d-2} = K^{d-2,\infty_1} = K^{d-2,\infty_2} \ . 
\label{Garyrule} \end{equation}
 
\noindent We will prove that, indeed, this is true for all integers $d \geq  5$ when 
$d$ is connected to the base field $K$ by equation (\ref{Eq1}). That squared 
overlap units from the lower dimension appear in eq. (\ref{kvadrater}) is a 
simple consequence of the prediction we have called 
Grassl's rule (\ref{Grasslrule}), because the order of the 
symmetry group of the aligned SIC is a factor of two larger than that of the SIC 
in the lower dimension while the ray class field (for these special baby overlaps) 
stays the same. 

As our example of a dimension where aligned SICs occur we choose $d = 35 = 
7\cdot 5$. This example is also interesting because it is the lowest dimension 
in which the Kopp--Lagarias inclusion lattice is non-trivial. See Figure 
\ref{fig:d35gitter}. In $d = 7 = 2^2 + 3$ there are two SICs with $f = 1$, 2, and overlap 
units as given in Table \ref{tab:1och2mod3}. In Table 
\ref{tab:35faser} the alignment makes itself felt in the first four rows,  
causing an otherwise unexpected behaviour among the baby overlaps there. 

When $f = 4$, $e = 2$, there is a slight 
problem in the first row. In this case the Stark 
units $S_{4\overline{\partial}}$ do not generate the field 
$K^{4\overline{\partial},\infty_1}$, and Grassl's 
rule for the exponents becomes a little ambiguous.  
Note also that there are a large number of identities 
connecting Stark units in the `small' fields, 
so the entries in Table \ref{tab:35faser} can be rewritten in various ways.

\section{The moduli $\mm_0 = $ $d-3$, $d-2$, $d-1$, $d$}
\label{sec:sec9}
Motivated by the notion of SIC alignment, 
we briefly develop some elementary class field theory results which would 
seem to be interesting in their own right, as well as in their applications to section~\ref{sec:sec8}. 
The finite part of the modulus~$\mm_0$ will be a $K$-integral ideal of $\ZK$ 
and $\mm_\infty$ an `infinite modulus': a 
possibly empty subset of $\{\infty_1,\infty_2\}$. 
The key equation (\ref{Eq1}) is in place, and we examine the degrees of the field extensions
over the totally real `base' ray class 
field~$K^{\mm_0}$ as we vary~$\mm_\infty$. 

Recall that we denote by $\tau$ the 
generator of the Galois group $\Gal{K}{\Q}$, 
and that~$\jj$ and~$\jj^\tau$ will 
denote the two embeddings of 
$K$ into~$\R$, chosen so that $\jj(\sqrt{D}) > 0$ 
(meaning $\jj$ is associated with the place denoted by $\infty_2$, 
as explained in section~\ref{sec:sec2}). 
The function $\nm x = \jj(x)\jj^\tau(x)$ denotes the 
field norm from $K$ down to~$\Q$. 
Recall that $u_D$ is defined to be the first 
totally positive power of $\uf$, in terms of $\uf$: 
\begin{equation}
u_D = 
\begin{cases}
	&	{\uf, \textrm{ when $\nm \uf = +1$;}} \\
	&	{\uf^2, \textrm{ when $\nm \uf = -1$.}} 
\end{cases}
\end{equation}
The fundamental unit $\uf$ is always assumed 
to have been chosen so that $\jj(\uf)>1$.  Hence 
$0 < \jj^\tau(\uf) < 1$ when $\nm \uf = 1$ and 
$-1 < \jj^\tau(\uf) < 0$ when $\nm \uf = -1$. 
By \eqref{dees} the dimensions $d_\ell(D)\geq4$ are 
given by $d=d_\ell(D)=u_D^\ell+u_D^{-\ell}+1$. 
The main result is as follows.

\begin{theorem}\label{hydrahead}
With notation as above, for any $d=d_\ell(D)\geq4$: 
\begin{enumerate}[label=(\Alph*)]
\item
let $\mm_0 = (d-1)\ZK$ or $(d-2)\ZK$ (this latter only when~$d\geq5$). 
Then $K^{\mm_0} = K^{\mm_0\infty_1} = K^{\mm_0{\infty_2}}$, and 
$K^{\mm_0\infty_1\infty_2}$ has degree $2$ over each of them. 

\item
let $\mm_0 = d\ZK$ or $(d-3)\ZK$ (this latter provided $d\neq4,5,8$). 
Then adding in a real place to $\mm_\infty$ always 
increases the degree of the field $K^{\mm_0\mm_\infty}$ by $2$. 
\end{enumerate}
\end{theorem}

In case (A) in particular, therefore, 
$K^{\mm_0\infty_1\infty_2}$ is 
a \emph{CM field}: that is to say, 
a totally complex field which is an 
extension of degree $2$ of a totally real field. 
Indeed, by lemma \ref{Urk1} \ref{sroots} 
we see that~$K^{\mm_0\infty_1\infty_2}$ 
is Galois over $\Q$ but that it also must contain non-real 
roots of unity. Hence the CM assertion will be true 
provided that we are able to show the first part, which since 
$K^{\mm_0}$ is by definition totally real, amounts to showing 
that~$K^{\mm_0\infty_1}$ and~$K^{\mm_0{\infty_2}}$ are 
also totally real. This in turn will be true (since $K^{\mm_0}$ is 
a subfield of both) if and only if the respective ray class groups 
are isomorphic. This is how we shall prove part~(A): see 
in particular the diagram~\eqref{twol}. 

Part~(B) is the situation described in Figure~\ref{fig:Nonmin2}; 
because the proof is just a straightforward modification 
of an argument in~\cite{AFMY}, we have moved 
it to appendix~\ref{sec:A4}, together with 
some technical components of the 
proof of part~(A). 

In order to put this in context, a first step is to outline what possible 
degrees {can} occur for these extensions. 
Let us free up $\mm_0$ to be \emph{any} 
finite ideal of $\ZK$ for a moment. We are interested 
in the sequence of degrees
\begin{equation}  \left[ \ [ K^{\mm_0 \infty_1} \colon K^{\mm_0} ], \  [ K^{\mm_0 \infty_2} \colon K^{\mm_0} ], 
\ [ K^{\mm_0 \infty_1\infty_2} \colon K^{\mm_0} ] \ \right] \ . \label{degz} \end{equation}
By standard class field theory, 
the first two of the indices in~\eqref{degz} must 
equal either $1$ or $2$, and by similar reasoning 
the third must be $1$, $2$ or~$4$; indeed 
$K^{\mm_0 \infty_1\infty_2}$ is restricted in all 
cases to be an extension of degree either $1$ 
or $2$ of each of its subfields~$K^{\mm_0 \infty_1}$ 
and $K^{\mm_0 \infty_2}$. 
The possible sequences are 
\begin{enumerate}[label=(\Roman*)]
\item $\nm{\uf}=1$, $\mm_0$ rational or irrational: $[1,1,2], \ [2,2,4]$;
\item $\nm{\uf}=-1$, $\mm_0$ rational: $[1,1,1], \ [1,1,2], \ [2,2,4]$;
\item $\nm{\uf}=-1$, $\mm_0$ irrational: $[1,1,1], \ [1,1,2], \ [1,2,2], \ [2,1,2], \ [2,2,4]$.
\end{enumerate}
The configurations~$[1,2,2]$ and~$[2,1,2]$ already played a 
role in section~\Ref{sec:sec4}: see 
eq.~(\ref{from16}) and ref.~\cite{ABGHM}. They can 
occur only for non-rational moduli, 
and only when the norm of the fundamental unit is~$-1$. 

However, for the purposes of theorem~\ref{hydrahead}, 
we may discard all but $[1,1,2]$ and $[2,2,4]$ since we shall 
always be choosing a modulus ideal which is 
principal over $\ZK$ and generated 
by a rational integer~$\geq3$, 
where the Galois involution $\tau\in\Gal{K}{\Q}$ 
cannot change the degrees of the intermediate extensions. 

Considering all moduli of the form 
$\mm_0 = (\lambda_1 + \lambda_2 d)\ZK$, 
for $\lambda_1$, $\lambda_2 \in \Z$ not both zero, 
the `normal' situation is that any such fixed form for the modulus yields 
a mixture of both outcomes in some ratio. 
So it is worth making the following observation: 

\begin{remark}
Expressed as $\Z$-linear functions of $d$ as above, 
the moduli $d-1$ and $d-2$ would appear 
to be \bf unique \it in giving $[1,1,2]$ for every $d$ 
relative to its associated field~$K$. 

On the other hand, it is relatively easy to find 
examples where the r\'egime is $[2,2,4]$, proving it 
using the same techniques as in the proof of (B). 
An illustration would be $(4d-21)\ZK$, for $d\geq6$.  
\end{remark}

The remainder of this section~\ref{sec:sec9} is devoted to a proof 
of part (A). We rely on results in appendix~\ref{sec:A4}. 

For general $\mm_0$ and $\mm_\infty$, 
define~$\U\mm$ to be the subgroup of $\ZKx$ whose elements are 
simultaneously congruent to~$1\bmod\mm_0$ and positive 
at all of the real places inside the set $\mm_\infty$ 
(this is often called \emph{multiplicative 
congruence}, and denoted by $\bmod^\times\mm$). 
The main technical ingredient for the proof is 
the exact sequence of global class field theory, as follows. 
See for example equation (2.7) of~\cite{cohenstev}. 
\begin{equation}\label{globcft}
1 \rightarrow \U{\mm} 
	\longrightarrow {\ZKx} 
		\xlongrightarrow{\psi } \mg{\mm_0} \times{\{\pm1\}}^{\# \mm_\infty } 
			\xlongrightarrow{\alpha} \mathrm{Gal}(K^\mm / K) 
				\longrightarrow \cK   
					\rightarrow 1.
\end{equation}
As we mentioned in section~\ref{sec:sec2}, the \emph{Artin map}~$\alpha$  
establishes an isomorphism between the ray 
class group of modulus $\mm$, and the Galois 
group of the ray class field $K^\mm$ over $K$. 

Given any $x$ in the global units $\ZKx$, the 
natural map $\psi$---whose kernel is the unit subgroup~$\U{\mm}$ 
just defined---is given by reduction 
of~$x$ modulo $\mm_0$ into the first component~$\mg{\mm_0}$, 
and then by the signs $\pm1$ induced 
according to each real embedding of~$x$ into the 
second component~${\{\pm1\}}^{\# \mm_\infty }$. 
For calculations we mention that the codomain 
of~$\psi$ is known as the \emph{ray residue ring} in Magma~\cite{Magma}. 
In particular this codomain is finite, which constrains the image of $\psi$ 
also to be finite: see~\cite[VI \S1]{Lang}. 

Just to place things in context, in section~\ref{sec:sec5} we made 
use of the Euler generalised totient function, 
which ties back to this sequence in that it is 
the order of the term~$\mg{\mm_0}$. 

The varying infinite modulus in theorem~\ref{hydrahead} (A) 
is captured by the following 
commutative diagram with exact rows, built via the 
functoriality of the Artin isomorphism from two versions of~\eqref{globcft} 
wherein we move from a given $\mm_\infty$ to a subset. 

For brevity we just annotate the maps in the 
top row with the representative 
subscript $\jj$ to indicate ramification being 
allowed at one real place, 
with no subscript denoting $\mm_\infty = \{\}$: 
\begin{equation}
\label{twol}
\begin{adjustbox}{width=\textwidth}
\begin{tikzcd}
1 \arrow[r] & 
\U{\mm_0,\jj} \arrow[r] \arrow[d,"\gamma"] &
\ZKx \arrow[r, "\psi_{\jj}"] \arrow[d,equal] &
\mg{\mm_0} \times \{\pm1\} \arrow[r, "\alpha_{\jj}"] \arrow[d,"\rho"] &
\mathrm{Gal}(K^{\mm_0,\jj} / K) \arrow[r] \arrow[d,"\mu"] &
\mathcal{C}_K \arrow[r] \arrow[d,equal] & 1 \\
1 \arrow[r] &
\U{\mm_0,\{\}} \arrow[r] &
\ZKx \arrow[r, "\psi"] &
\mg{\mm_0} \arrow[r, "\alpha"] &
\mathrm{Gal}(K^{\mm_0,\{\}} / K) \arrow[r] &
\mathcal{C}_K \arrow[r] &
1
\end{tikzcd}
\end{adjustbox}
\end{equation}
The maps $\gamma$, $\rho$ and $\mu$ are implicitly 
defined by the diagram. 
By the definition of the~$\U\mm$ (or by the snake lemma) we know 
that~$\gamma$ is an injection, since~$\U{\mm_0\mm_\infty}$ 
is in all cases just~$\U{\mm_0,\{\}}$ under an additional compatible restriction. 
Similarly $\rho$ is a projection onto the component corresponding 
to the finite part $\mm_0$ of the modulus and so is surjective, with 
kernel~$\ker\rho \cong \{\pm1\}$ representing 
the signs under the single real embedding $\jj$. 
From what was said above we need to prove that $\mu$ is an 
isomorphism whenever $\jj \in \infty_1$ or $\jj \in \infty_2$. 

The key ingredient in the proof of 
theorem~\ref{hydrahead}~(A) is the following. 

\begin{proposition}\label{coeur}
Whenever $\mm_0$ is either $(d-1)\ZK$ 
or~$(d-2)\ZK$ and $\infty = \infty_1$ or~$\infty_2$, 
the cokernel of the natural inclusion 
$
\gamma \colon \U{\mm_0\infty}\ \lhra \ \U{\mm_0\{\}}
$
has order~$2$. 
\end{proposition}
\noindent
Although we prove it formally below using some 
elementary homological algebra, it is relatively straightforward to see from 
diagram~\eqref{twol} that this defect of 2 is precisely what is needed 
to cancel out the extra sign module $\{\pm1\}$ in order to ensure 
that the Galois groups at the right hand side---given that the 
map between the class groups $\cK$ is the identity---have the same order. 

\begin{proof}[Proof of proposition \ref{coeur}]
By splitting \eqref{twol} in the usual way into commutative 
sub-diagrams each with two rows of short exact 
sequences, we may apply the snake lemma~\cite[(II.28), p.120]{gelman} 
to each of them.  From the left-hand part we deduce: 
\begin{equation}
\begin{tikzcd}\label{tiki}
1 \arrow[r] & 
\U{\mm_0,\jj} \arrow[r] \arrow[d,hook,"\gamma"] &
\ZKx \arrow[r, "\psi_{\jj}"] \arrow[d,equal] &
\ker\alpha_{\jj} \arrow[r] \arrow[d,two heads,"\overline{\rho}"] & 
1 \\
1 \arrow[r] &
\U{\mm_0,\{\}} \arrow[r] &
\ZKx \arrow[r, "\psi"] &
\ker\alpha \arrow[r] & 
1
\end{tikzcd}
\end{equation}
where we have defined~$\overline{\rho} \colon \img\psi_{\jj} \lra \img\psi$, 
which we view as 
$\overline{\rho} \colon \ker\alpha_{\jj} \lra \ker\alpha$ 
by exactness, as the appropriate restriction of $\rho$. 
Straight away we see that $\overline{\rho}$ is 
surjective (as indicated by the standard double-headed arrow; note 
that similarly we use the standard hooked arrow to indicate injections), 
and most importantly that: 
\begin{align}\label{kerker}
\ker\overline{\rho}	\cong	\coker\gamma		
			= 	\U{\mm_0,\{\}} / \U{\mm_0\jj} .
\end{align}
From now onwards let us denote a finite abstract cyclic group of order $N$ by~$C_N$. 
We wish to show that $\coker\gamma \cong C_2$; so it 
is enough to show that~$\ker\overline{\rho}$ has order~$2$. 
The corresponding right-hand commutative exact sub-diagram 
looks like this: 
\begin{equation}
\begin{tikzcd}\label{ritiki}
1 \arrow[r] & 
\img\alpha_{\jj} \arrow[r,"\alpha_{\jj}"] \arrow[d,"\overline{\mu}"] &
\mathrm{Gal}(K^{\mm_0,\jj} / K) \arrow[r] \arrow[d,"\mu"] &
\cK \arrow[r] \arrow[d,equal] & 1 \\
1 \arrow[r] &
\img\alpha \arrow[r,"\alpha"] &
\mathrm{Gal}(K^{\mm_0,\{\}} / K) \arrow[r] &
\cK \arrow[r] & 1 ,
\end{tikzcd}
\end{equation}
where again the notation $\overline{\mu}$ is understood 
to represent $\mu$ restricted to the smaller domain 
$\img\alpha_{\jj} \trianglelefteq \mathrm{Gal}(K^{\mm_0,\jj} / K)$. 

With this in mind we may finally draw the `central' diagram of short exact 
sequences (ie not including the ideal class groups at the right): 
\begin{equation}\label{zen}
\begin{adjustbox}{scale=1}
\begin{tikzcd}
1 \arrow[r] & 
\ker\overline{\rho} \arrow[r] \arrow[d,hook] &
\ker\rho = \{\pm1\} \arrow[r] \arrow[d,hook] &
\ker\overline{\mu} \arrow[r] \arrow[d,hook] & 
\ldots \\
1 \arrow[r] & 
\ker\alpha_{\jj} \cong \img\psi_\jj \arrow[r] \arrow[d,"\overline{\rho}", two heads] &
\mg{\mm_0} \times \{\pm1\} \arrow[r, "\alpha_{\jj}"] \arrow[d,"\rho", two heads] &
\img\alpha_{\jj} \arrow[r] \arrow[d,"\overline{\mu}"] & 
1 \\
1 \arrow[r] &
\ker\alpha \cong \img\psi \arrow[r] &
\mg{\mm_0}  \arrow[r, "\alpha"] &
\img\alpha \arrow[r] & 
1.
\end{tikzcd}
\end{adjustbox}
\end{equation}
The top row here is from yet another application of 
the snake lemma. 
Since by construction~$\rho$ is onto, $\overline{\mu}$ 
must also be onto and~$\overline{\rho}$ 
is surjective from above. 
So the only possibilities---which are mutually exclusive---are that 
\begin{align}\label{possum}
\ker\overline{\rho} \cong C_2 \textrm{ and}\ker\overline{\mu} = 1;& \textrm{ or} \nonumber\\
\ker\overline{\mu} \cong C_2  \textrm{ and}\ker\overline{\rho} = 1. &
\end{align}
But we know from lemma~\ref{fish} (see there for the 
definition of~$w$) that 
$\img\psi_\jj \cong C_w \times C_2$ and that 
$\img\psi \cong C_w$: just by counting orders, 
this forces $\ker\overline{\rho}$ to have order $2$, as required. 
\end{proof}

This completes the proof of 
assertion (A) of theorem \ref{hydrahead}, as follows. 
The first part---namely, that 
$K^{\mm_0} = K^{\mm_0\infty_1} = K^{\mm_0{\infty_2}}$---is 
that in our situation, the map $\mu$ is always an isomorphism. 
But by \eqref{possum}~$\overline{\mu}$ is injective and by~\eqref{zen} 
it was onto, so~$\overline{\mu}$ is an isomorphism; 
by the snake lemma applied to~\eqref{ritiki}, so is~$\mu$. 
The second part---namely that relative to these three 
identical totally real fields the 
full field~$K^{\mm_0{\infty_1\infty_2}}$ is indeed a 
proper extension of degree~2---follows from lemma~\ref{Urk1} \ref{sroots}. \qed

\section{Comments in closing}\label{sec:sec10}
In closing we can say that we have verified the four claims made about overlap units in Section \ref{sec:sec1}, 
for all our examples. 
The fact that the overlaps often sit in a proper subfield of $K^{fd \infty_1}$ is no objection, since we are in 
effect taking relative norms to reach that subfield. 

Using Hilbert space arguments it can be shown that there are cases where some squares of overlap units necessarily 
are equal to $+1$. This is the exceptional case where we can support our claims with proofs, because these 
overlap units appear precisely when Stark's construction results in units equal to $+1$. The behaviour of the 
ray class fields with finite modulus $d-3$, $d-2$, $d-1$, and $d$, as expounded in section \ref{sec:sec9}, may 
be of independent interest.   

We did find a few examples (when $d = 12$, 28, and 35) where the Stark units do not generate the field 
in which they sit. Other examples are known. This phenomenon causes no special problem for the 
construction of the SIC, although the rule (\ref{Grasslrule}) for the exponents of the Stark units 
can become a delicate one to apply.  

We were able to offer only the vaguest of rationales for why it is square roots of Stark units that appear, 
and why one needs to include products of square roots of Stark units from several distinct fields in the 
non-minimal case. Indeed, we have been concerned with {\it how} Stark units enter SIC overlaps, not 
{\it why} they do so. 

An important point to be raised is the relation between the facts we have observed on the one hand, 
and the conjectural identification of all the overlap units with special values of the Shintani--Faddeev modular 
cocycle on the other \cite{Kopp3, Kopp4}. In principle it should be possible to derive all our results from the 
main theorem (Theorem 1.1) of ref. \cite{Kopp3}. This remains a task for the future however. 

\vspace{1cm}

\noindent {\bf Acknowledgements:}

\

\noindent We received indispensable help from Marcus Appleby and Markus Grassl. 
IB is grateful to the taxpayer for support through the Digital Horizon Europeproject FoQaCiA, GA
No. 101070558, funded by the European Union, NSERC (Canada), and UKRI
(UK); and to the Mathematics Department at Stockholm University for being allowed to 
use their computers. 
GM is grateful to Myungshik Kim, Terry Rudolph and the QOLS group 
at Imperial College for their ongoing hospitality. 

\vspace{1cm}

\appendix

\section*{Appendices}

\section{The group}\label{sec:A1}

\noindent The Weyl--Heisenberg group $H(d)$ is a finite dimensional analogue 
of the group that underlies Heisenberg's uncertainty relation \cite{Weyl}. 
It is a central extension of the direct product of two cyclic groups, 
and can be presented as 

\begin{equation} ZX = \omega XZ \ , \hspace{5mm} X^d = Z^d = \omega^d = 
{\bf 1} \ , \end{equation}

\noindent with the understanding that $\omega$ commutes with everything 
and can be represented as multiplication with a primitive root of unity. 
We choose 

\begin{equation} \omega = e^{\frac{2\pi i}{d}} \ . \end{equation}

\noindent If $d$ is even the centre is enlarged to include also 

\begin{equation} \tau = - e^{\frac{\pi i}{d}} \ . \end{equation} 

\noindent Once this has been agreed on there is an essentially unique 
irreducible unitary representation in dimension $d$, given by 

\begin{equation} Z \lvert r\rangle = \omega^r \lvert r\rangle \ , \hspace{5mm} 
X \lvert r\rangle =  \lvert r + 1\rangle \ . \end{equation}

\noindent The basis vectors are indexed by integers modulo $d$. 
Throughout this paper we assume that this representation is 
used, unless $d = 4$ in which case a variation is useful \cite{monomial}. 
An important fact is that if $d = d_1d_2$ where $d_1$ and $d_2$ are 
relatively prime then there is a canonical isomorphism  
\cite{monomial, Irina}

\begin{equation} H(d) = H(d_1) \times H(d_2) \ . \end{equation}

\noindent To show this the Chinese remainder theorem is used. It means 
that once we understand the group in prime power 
dimensions everything else follows. 

It is useful to introduce the $d^2$ displacement operators 

\begin{equation} D_{i,j} = \tau^{ij}X^iZ^j \ , \end{equation}

\noindent with this precise choice of the prefactor \cite{Marcus}. For 
them the group law reads 

\begin{equation} D_{i,j}D_{k,l} = \tau^{jk - il}D_{i+k,j+l} \ . \end{equation}

\noindent The symplectic form that appears in the exponent on the right 
hand side is the key to understanding the unitary automorphism group of 
the Weyl--Heisenberg group, and the symmetries exhibited by its orbits. 
In the representation we are using, there are outer automorphisms of the 
Weyl--Heisenberg group that act by conjugation, forming a representation 
of the symplectic group $SL(2,{\mathbb{Z}}/d{\mathbb{Z}})$. Hence we can associate 
every symplectic matrix $G$ with a unitary operator $U_G$ according to the 
following scheme: 

\begin{equation} \left( \begin{array}{c} i' \\ j' \end{array} \right) = 
\left( \begin{array}{cc} \alpha & \beta \\ \gamma & \delta \end{array} \right) 
\left( \begin{array}{c} i \\ j \end{array} \right) \hspace{5mm} 
\leftrightarrow \hspace{5mm} U_GD_{i,j}U_G^{-1} = D_{i',j'} \ . 
\label{symplektiskt} \end{equation}

\noindent If $\alpha \delta - \beta \gamma = - 1$ mod $d$ the operator $U_G$ is 
an anti-unitary rather than a unitary operator. Moreover, the integer entries of $G$ 
should be taken modulo $2d$ if $d$ is even; but this is not important for us 
\cite{Marcus}. 

The unitary or anti-unitary operators $U_G$ transform SICs to SICs. Symmetries 
of SICs always include a Zauner unitary, that is to say an operator $U_G$ coming 
from a symplectic matrix of order 3 and trace $-1$. Unless $d>3$ is divisible by 
3 but not by 9 there is a unique conjugacy class of such matrices. 

Another key property of the group is that it forms a unitary operator basis 
\cite{Schwinger}. This means that any operator acting on ${\mathbb{C}}^d$ can 
be expanded as 

\begin{equation} A = \frac{1}{d}\sum_{i,j = 0}^{d-1} 
D_{i,j}\mbox{Tr}AD_{-i,-j} \ . \label{Schwinger} \end{equation}

\noindent We call this the {\it Schwinger formula} for ease of reference. If we 
choose $A =  \lvert \Psi_0 \rangle \langle \Psi_0  \lvert $ we learn that the overlaps 
determine the fiducial vector uniquely up to an irrelevant phase. 

Now we return to the beginning and observe that the choice $\omega = 
e^{2\pi i/d}$ was arbitrary. Any primitive $d$th root of unity would serve. 
Such a change of representation is achieved by the Galois transformation 
$\omega \rightarrow \omega^a$, and leads us to bring in the group 
$GL(2,{\mathbb{Z}}/d{\mathbb{Z}})$ \cite{AYAZ}. This needs an extra generator 
transforming the displacement operators as 

\begin{equation} G = \left( \begin{array}{cc} 1 & 0 \\ 0 & a \end{array} 
\right) \ , \hspace{5mm} X \rightarrow X \ , \hspace{5mm} 
Z \rightarrow Z^a \ , \end{equation}

\noindent where $a$ is any integer invertible modulo $d$. 
Anti-unitary transformations are a special case of this. 

The reader may ask if maximal sets of equiangular vectors that are not 
orbits of the Weyl--Heisenberg group can exist. The answer is `no' if 
$d = 2$, 3 \cite{Lane, Feri} and `probably no' if $d = 4$, 5 
\cite{Samuel}. In dimension $d = 8$ an example that is an orbit of a 
differently defined Heisenberg group does exist \cite{Blake}. 
Possibly this is the only such example, but the question is open. 

\section{Order four symmetries}\label{sec:A2}

\noindent The prominent role of order three symmetries 
in Section \ref{sec:sec3} {\it et passim} is 
remarkable. One can ask if, say, a symmetry of 
order four would lead to something of interest? Assume 
for simplicity that $d = p$ is a prime. A first 
observation is that the centraliser of an order four symplectic 
matrix within $GL({\bf Z}/d{\bf Z})$ has 
order $p^2-1$ if $d = 3$ mod 4 and order $(p-1)^2$ if $d = 1$ mod 4. 
So far, so good. If we want to mimic the 
constructions made for SICs the next step is to 
ask if, given $d$, it is possible to choose a real 
quadratic field $K$ such that the ray class field $K^{d\infty_1}$ 
has a degree equal to the order of the centraliser 
divided by 4, or possibly divided by 8 should 
anti-unitary symmetries occur. Will this lead to a 
magical formula differing from equation (\ref{Eq1}), 
and if so what kind of Weyl--Heisenberg orbits 
can be obtained from fiducial vectors constructed 
using such a field?  But this seems to be a dead 
end. If we also insist that the degree increases by 
a factor of 2 when an infinite place is added to the 
modulus there are only a very few exceptional 
cases (such as $p = 3$) where such fields exist. 
This is one way to see why symmetries of order 
three are very special.

\section{Finding accessible examples}\label{sec:A3}

\begin{table}[h]
\caption{{\small Odd dimensions where non-minimal SICs occur, and their excess. 
The excess is underlined if overlaps given by Stark units have been 
computed. For the bracketed cases only baby overlaps were computed.}}
  \smallskip \smallskip
\hskip 0.0cm
{\renewcommand{\arraystretch}{0.8}
\begin{tabular}
{|c|ccccccccccccccc|}\hline 
$f$ & {\small 7} & {\small 11} & {\small 15} & {\small 17} & {\small 19} & 
{\small 21} & {\small 23} & {\small 27} & {\small 31} & {\small 35} & 
{\small 39} & {\small 43} & {\small 47} & {\small 49} & {\small 51} \\
\hline
{\small 2} & {\small \underline{1}} & {\small \underline{1}} & 
{\small \underline{1}} & & {\small \underline{1}} & & 
{\small \underline{1}} & {\small 1} & {\small \underline{1}} & 
{\small \underline{1}} & {\small (\underline{1})} & {\small 1} & {\small 1} & 
& {\small 1} \\ 
{\small 3} & & & & {\small \underline{1}} & & {\small (\underline{3})} & & & & 
{\small \underline{1}} & {\small (\underline{3})} & & & & \\ 
{\small 4} & & & {\small \underline{1}} & & {\small \underline{2}} & & & & 
{\small (\underline{1})} & {\small (\underline{2})} & & & {\small 1} & & 
{\small 2} \\ 
{\small 5} & & & & & & & & & & & & & & {\small 2} & \\ 
{\small 6} & & & & & & & & & & {\small \underline{1}} & {\small 3} & & & & \\ 
{\small 8} & & & & & {\small \underline{2}} & & & & & & & & {\small 2} & & \\ 
{\small 12} & & & & & & & & & & {\small (\underline{2})} & & & & & \\ 
\hline 
\end{tabular}
}
\label{tab:numerologi}
\end{table}

\begin{table}[h]
\caption{{\small Even dimensions where non-minimal SICs occur, and their excess.}}
  \smallskip \smallskip
\hskip 0.0cm
{\renewcommand{\arraystretch}{0.8}
\begin{tabular}
{|c|ccccccccccccccc|}\hline 
$f$ & {\small 8} & {\small 12} & {\small 24} & {\small 26} & {\small 28} & 
{\small 30} & {\small 44} & {\small 48} & {\small 52} & {\small 62} & 
{\small 66} & {\small 74} & {\small 78} & {\small 80} & {\small 84} \\ 
\hline 
{\small 3} & {\small \underline{1}} & {\small (\underline{3})} & & 
{\small (\underline{1})} &  & {\small 3} & {\small 1} & {\small (\underline{3})} 
& & {\small 1} & {\small 3} & & & {\small 1} & {\small (\underline{3})} \\ 
{\small 5} & & & {\small 2} & & {\small (\underline{4})} & & & & & & & {\small 2} & 
{\small 4} & & \\ 
{\small 7} & & & & & & & & {\small (\underline{3})} & {\small 6} & & & & & & \\ 
{\small 9} & & & & & & & & & & & & & & {\small 3} & {\small 9} \\ 
{\small 21} & & & & & & & & {\small 9} & & & & & & & \\ 
\hline 
\end{tabular}
}
\label{tab:numerologi2}
\end{table} 

\noindent In Tables \ref{tab:numerologi}--\ref{tab:numerologi2} 
we give two lists of low dimensions 
where non-minimal SICs occur, and their excess as defined in equation 
(\ref{excess}). The lists are complete up to the highest dimensions 
cited. If an underlined entry is not mentioned in the main text it is 
because there were no surprises there.

\section{Lemmas needed for theorem \ref{hydrahead}}\label{sec:A4}
In our situation we need to understand the image $\img\psi$ 
in the exact sequence~\eqref{globcft}, under certain conditions on $\mm$. 
It will be convenient to regard the units~$\ZKx$ of $K$ as 
decomposing in the following standard fashion, where we denote 
by~$<x>$ the (possibly infinite) multiplicative cyclic group 
generated by an element~$x$. 
The torsion is just $\pm1$ and so with our fixed 
choice of a fundamental unit $\uf$ we may write:
\begin{equation}\label{gengen}
\ZKx \ \ = \ \ <\uf> \times <-1> \ \ \cong \ \ \Z \times C_2 .
\end{equation}
We now begin to gather up some facts about the various 
invariants here. 
We continue to regard~$\mm_0$ as one of the 
two possibilities~$(d-1)\ZK$ or~$(d-2)\ZK$. 

Define $w$ to be the multiplicative order of the 
image of the fundamental unit~$\uf$ under reduction 
modulo the finite part~$\mm_0$ of the modulus, 
\emph{irrespective of the real places}. 
In other words, 
$w$ is the lowest positive integer such that 
$(\uf+\mm_0)^w = \uf^w + \mm_0 = 1 + \mm_0$ 
in~$\left(\ZK/\mm_0\right)^\times$. 
In order to avoid over-burdening the reader with notation 
we shall use this abbreviation $w$, assuming it known from the context 
whether $\mm_0$ be generated by $d_\ell-1$ or by $d_\ell-2$. 
The prevailing value of $\ell$, of course, is also contextual. 

\begin{lemma}\label{Urk1}
With $K$, $\ZK$ as above, choose some $s\in\Z$,~$s\geq3$ 
and write~$\mm_0 = s\ZK$. 
Let $\mm_\infty$ denote any subset of the real places 
$\{\infty_1,\infty_2\}$ of $K$. 
Then 

\begin{enumerate}[label=(\roman*), ref=(\roman*)]
\item\label{won}
$\U{\mm_0\mm_\infty}$ is a 
$\Z$-rank one torsion-free abelian group. 
In particular, therefore, 
t must have a generator 
of the form $\pm \uf^j$, where $j \in \{ \frac{w}{2}, w \}$. 
\item\label{sroots}
The ray class field $K^{\mm_0\infty_1\infty_2}$ 
contains the $s$-th roots of unity. 
\end{enumerate}
\end{lemma}

\begin{proof}\ref{won}
For the fact that it has free $\Z$-rank exactly 
equal to that of $\ZKx$---in this case~1, 
by Dirichlet's unit theorem---see \cite[VI \S1]{Lang}, as we pointed 
out above in the discussion after~\eqref{globcft}. 
So we only need show it has no torsion. 
But the roots of unity $\mathbf{\mu}(K)$ contained in $K$ 
are just $\{ \pm 1 \}$; so we only need verify 
that~$-1 \notin \U{\mm_0\mm_\infty}$: and this 
is ensured by the choice of~$s$ as a rational integer~$\geq3$. 
The second assertion is then a 
straightforward consequence of the decomposition 
in~\eqref{gengen} and the minimality of $w$. 

\ref{sroots} See the proof of proposition~9(i) 
in~\S4.2 of \cite{AFMY}. 
\end{proof}

In~\cite{AFMY}, where the finite part of 
the modulus~$\mm_0$ is~$d_\ell$ or~$2d_\ell$, 
the order~$w$ of $\uf$ modulo $\mm_0$ 
is~\emph{a priori} difficult to predict and requires a direct calculation for each case. 
However in the main cases of interest here---that is to 
say, $\mm_0 = (d_\ell-1)\ZK$ or $(d_\ell-2)\ZK$---it is 
always given by the following simple 
rules (parts \ref{deekm1} and~\ref{deekm2}), 
which we shall need in the proofs below. 
This is the key to why these cases give a simple 
hierarchy of ray class fields. 

We also include some other necessary technical observations.

\begin{lemma}\label{bells}
With notation as above, writing $d_\ell = d_\ell(D) \geq 4$, 
\begin{enumerate}[label=(\alph*)]
\item 
The order of $u_D$ modulo $(d_\ell-1)\ZK$ is $4\ell$ 
and $u_D^{2\ell}\equiv-1\bmod(d_\ell-1)\ZK$. 
\label{deekm1}
\item 
Assume $d_\ell\geq5$. 
The order of $u_D$ modulo $(d_\ell-2)\ZK$ is $6\ell$ 
and $u_D^{3\ell}\equiv-1\bmod(d_\ell-2)\ZK$. 
When $d_\ell=4$ the characteristic is~2 and 
the order is just~3. 
\label{deekm2}

Furthermore, 

\item
The order of $\uf$ modulo either of the moduli $\mm_0$ 
in \ref{deekm1}, \ref{deekm2} 
is the same as that of $u_D$, 
multiplied by $2$ if $\nm\uf=-1$. 
\label{youtoo}
\item
$w$ is even. \label{weven}
\item 
Suppose that $\nm \uf = -1$, the generator for~$\mm_0$ is $\geq3$ 
and that $\uf^\frac{w}{2} \equiv -1 \bmod \mm_0$. 
Then $\frac{w}{2}$ is even. \label{evev}
\end{enumerate}
\end{lemma}

\begin{proof}
By our choice of $d\geq4$, 
the multiplicative groups $\mg{d-1}$ 
and $\mg{d-2}$ are non-empty, so that any 
global unit in $\ZK$ maps to a non-zero 
invertible element in each of those rings. 
Hence for \ref{deekm1}, observe that by the 
definition of~$d = d_\ell$, 
\[
1+u_D^{2\ell} = {u_D^\ell}\frac{1+u_D^{2\ell}}{u_D^\ell}  =  {u_D^\ell}(d_\ell-1) \in {u_D^\ell}(d_\ell-1)\ZK =  (d_\ell-1)\ZK,
\]
or, as claimed, 
$u_D^{2\ell} \equiv -1 \bmod (d_\ell-1)\ZK$. 
We use this to prove the assertion about the order. 

Suppose that some smaller divisor $t$ 
of $4\ell$ were the order of $u_D$. 
(That is to say, the powers of $u_D$ modulo $\mm_0$ 
have already passed $1$ on the way to $u_D^{2\ell} \equiv -1$). 
By what we have just proven,~$t$ 
cannot divide into~$2\ell$; but it must divide into~$4\ell$. 
So $t=4m$ for some $m \lvert \ell$, $m<\ell$. 
Write $\ell = \nu m$ defining another positive integer $\nu>1$. 
We claim, in the first instance, that $\nu$ must equal~$3$. 
We shall invoke the standard floor, ceiling and nearest integer notation 
respectively $\lfloor x \rfloor$, $\lceil x \rceil$ and $[x]$ 
applied to any real number $x$. 

If $\nu=2$ then $t=2\ell$; on the other hand if $\nu=4$ then 
$t=\ell$; in both cases~$t\lvert2\ell$, a contradiction. 
Now, by the definition of~$t$ and by~\eqref{dees},
~$d_t = u_D^t + u_D^{-t} + 1 \equiv 3 \bmod (d_\ell-1)$, 
hence there exists some $r\in\Z$ such that 
\begin{equation}\label{rsol}
d_t := d_{4m} = rd_{\nu m} - r + 3. 
\end{equation}
By sheer size considerations since $d_\ell\geq4$ we see that $r$ must be positive. 
If $r=1$ then $d_{4m} = d_{\nu m} + 2$; suppose for 
a moment that~$\nu = 5$. 
Then since 
$d_{4m} = [u_D^{4m}+1]$,  
$d_{5m} = [u_D^{5m}+1]$ and 
$u_D \geq \frac{3+\sqrt{5}}{2} \approx 2.618...$ is 
the smallest possible value of the 
maximum fundamental unit for any $D$---and 
using the recurrence relation for the 
modified Chebyshev polynomials of the first kind, denoted~$T_{r}^*(X)$, 
in the proof of proposition~7 of~\cite{AFMY} for the 
strict inequality in the middle---the last equality 
coming from~\eqref{rsol} with our assumed~$r=1$:
\begin{equation}
2^m d_{4m} < \lfloor u_D \rfloor^m d_{4m} < d_{5m} < d_{5m}+2 = d_{4m},
\end{equation}
a contradiction for any $m\geq1$. 
The same argument is true with yet stronger reasoning 
for any $\nu\geq6$ or indeed any~$r\geq2$. 

Hence as claimed, the only possible value for $\nu$ is~$3$. 
We now proceed to rule that out as well. 
So suppose that $\nu=3$. 
The above polynomials~$T_{r}^*(X)$ operate 
to express $d_{rs}$ as a function of $d_s$ in \S3 of~\cite{AFMY}, yielding the 
following expressions in $d_m$: 
\begin{equation}\label{d3m}
d_{3m} = d_m^3 - 3d_m^2 + 3,
\end{equation}
and 
\begin{equation}\label{d4m}
d_{4m} = d_m^4 - 4 d_m^3 + 2 d_m^2 + 4 d_m;
\end{equation}
and hence we may deduce from~\eqref{rsol} 
the following rational integer congruence:
\begin{equation}\label{congdm}
2r \equiv -3 \bmod d_m.
\end{equation}
Once again we observe from the relative sizes in~\eqref{rsol} 
that $r$ is positive and of the order of $d_m = 1+[u_D^m]$, 
which translates in \eqref{congdm} to $r = \frac{d_m-3}{2}$. 
This by the way forces $d_m$ to be odd, because $r\in\N$. 

Now substitute this in turn back into \eqref{rsol}, this 
time reading it modulo $d_m^2$, again in the light of 
\eqref{d3m} and~\eqref{d4m}. We see that 
$$
4d_m \equiv d_m \bmod d_m^2,
$$
which is attainable only at $d_m=4$, which is 
even, a contradiction. 

For \ref{deekm2}, we again observe that by definition of the $d_\ell$, 
\[
u_D^{2\ell} - u_D^\ell + 1 = {u_D^\ell}\frac{u_D^{2\ell} - u_D^\ell + 1}{u_D^\ell}  =  {u_D^\ell}(d_\ell-2)	\in {u_D^\ell}(d_\ell-2)\ZK = (d_\ell-2)\ZK;
\]
in particular, multiplying by the non-zero element $1+u_D^\ell$ gives 
${u_D^{3\ell} + 1} \in (d_\ell-2)\ZK$ 
and so once more as claimed,~$u_D^{3\ell} \equiv -1 \bmod (d_\ell-2)\ZK$. 

The only possibility we have left open is that 
the zero-divisor~$1+u_D^\ell$ is in fact zero in the 
ring $\ag{d_\ell-2}$; that is, that~$u_D^\ell \equiv -1$. 
But $d_\ell - 2 = u_D^\ell + u_D^{-\ell} - 1$ and so this would say that 
$-3 \equiv 0 \bmod (d_\ell-2)$, which is clearly wrong 
unless $d_\ell=1,3$ or $5$, the first two of which we have excluded 
and the third case of which may instead be verified by direct calculation. 
The minimality of $6\ell$ follows from 
the same elimination argument as for $4\ell$ above, this time 
with~$\nu=5$ being the only slightly tricky case. 
The assertion about the anomalous case $D=5,\ell=1,d=4$ 
is a simple calculation, which boils down to the fact that 
the order of the multiplicative group~$\mg{2}$ is just~$3$. 

For~\ref{weven}: the order $w$ is just the 
case $\ell=1$ from the results just proved, 
multiplied by $2$ if $\nm\uf = \uf\uf^\tau = -1$. 
Hence in particular it is always even. 
Moreover~\ref{youtoo} is clearly true provided that 
$\Phi(\mm_0) = \#\mg{\mm_0}$ is even and 
$\uf + \mm_0 \neq 1 + \mm_0$, both of 
which follow from the results just proven. 
We use~\ref{weven} to prove \ref{evev}, 
since the expression is valid as $w$ is even: 
\begin{align}
\uf^\frac{w}{2} \equiv -1 \bmod \mm_0 
	& \implies (\uf^\frac{w}{2})^\tau \equiv (-1)^\tau \equiv -1 \bmod \mm_0 \nonumber \\
&\implies \uf^\frac{w}{2} (\uf^\frac{w}{2})^\tau \equiv 1 \bmod \mm_0 \nonumber \\
&\implies (\uf \uf^\tau)^\frac{w}{2} \equiv 1 \bmod \mm_0 \nonumber \\
&\implies (\nm \uf)^\frac{w}{2} = (-1)^\frac{w}{2} \equiv 1 \bmod \mm_0, 
\end{align}
by our assumption that $\uf \uf^\tau = \nm \uf = -1$, and 
so $\frac{w}{2}$ is indeed even since the characteristic of 
$\ag{\mm_0}$ is $\geq3$ in all cases except $d=4$, $\mm_0 = (d-2)\ZK$. 
\end{proof}

Lemma \ref{bells} \ref{evev} is now the remaining 
technical ingredient needed to 
delineate the two cases which can occur for our two possible 
choices of finite moduli~$\mm_0$, namely $(d_\ell-1)\ZK$ and 
$(d_\ell-2)\ZK$. Recall that we write~$\mm = \mm_0\mm_\infty$. 

\begin{lemma}[]\phantomsection\label{fish}
\begin{enumerate}[label=(\alph*)]
\item\label{dag} There exists $r\in\N$ such that 
$\uf^r \equiv -1 \bmod^\times \mm_0\mm_\infty$ 
if and only if $\uf^\frac{w}{2} \equiv -1 \bmod \mm_0$ 
and $\mm_\infty = \{\}$. 
In particular, $\ker\psi = < - \uf^\frac{w}{2} >$ and 
$\img\psi \cong C_w$. 
\label{obvev}
\item\label{dagdag} Suppose that 
$\uf^\frac{w}{2} \not\equiv -1 \bmod^\times \mm_0\mm_\infty$. 
Then $\ker\psi = \U{\mm} = < \uf^w >$ for all four 
infinity types $\mm_\infty = \{\},\{\infty_1\},\{\infty_2\},\{\infty_1,\infty_2\}$ and consequently 
$\img\psi \cong C_w \times \{ \pm1 \}$.  \label{odev}
\end{enumerate}
\end{lemma}

\begin{proof}
\ref{obvev}: 
$\implies\colon$ The hypothesis implies the existence of a minimal positive 
integer $r$ such that $\uf^r \equiv -1 \bmod^\times \mm$. 
By the minimality of both $w$ and $r$, indeed 
$w = \gcd(2r,w) = 2r$ is even, and $r=\frac{w}{2}$. 
Therefore $-\uf^\frac{w}{2} \equiv 1 \bmod^\times \mm$, and 
consequently must be positive at any real places in~$\mm_\infty$. 

But $\uf>0$ under~$\infty_2$. Hence when $\nm \uf = +1$, $\uf$ is positive at 
both infinite places; so in particular~$-\uf^\frac{w}{2} < 0$ 
and therefore it must be the case that $\mm_\infty=\{\}$. 
On the other hand when $\nm \uf = -1$, the 
possibilities that~$\mm_\infty = \{\infty_2\}$ or $\{\infty_1,\infty_2\}$ 
are excluded for the same reasons. 
Finally, when the norm is~$-1$ and $\mm_\infty=\{\infty_1\}$, 
we know from lemma~\ref{bells} \ref{evev} that $\frac{w}{2}$ is even, 
and we would need $-\uf^\frac{w}{2} > 0$ at $\{\infty_1\}$, 
which is impossible because $\uf^2$ is totally positive. 
This once again only leaves the possibility that $\mm_\infty = \{\}$. 
So~$-\uf^\frac{w}{2}$ generates $\U{\mm_0,\{\}}$.  
Note by the way that, curiously, when $\nm \uf = 1$ we 
are agnostic as to whether $\frac{w}{2}$ 
is itself even or odd. 

That the image of $\psi$ is cyclic follows from 
the fact that $w$ is even and the 
kernel is generated by a negative element. 
Notice that this effectively says that arithmetically within the 
ray residue ring, $\uf^\frac{w}{2}$ ``is'' $-1$. 

$\impliedby\colon$~$w$ is even by lemma~\ref{bells}~\ref{weven} so set $r = \frac{w}{2}$. 

\ref{odev}: 
Again, $w$ is even by lemma~\ref{bells} \ref{weven} and so 
the expression~$\uf^\frac{w}{2}$ is well-defined. 
But the argument in the proof of part~\ref{obvev} 
just now shows that by the minimality of~$w$, there 
is \emph{no} $j$ such that $\uf^{j} \equiv -1 \bmod \mm_0$. 
Hence it must be the case that~$\uf^\frac{w}{2}$ is some 
\emph{other} square root of $1$ modulo $\mm_0$, 
and so its negative $-\uf^\frac{w}{2}$ must fail to be $1$ 
(respectively, fail to be positive) at some 
finite (respectively, infinite) place dividing $\mm$. 
Consequently it is excluded from the kernel of~$\psi$. 
So by elimination~$\uf^w$ generates the kernel of $\psi$. 
\end{proof}

\paragraph{Proof of Theorem \ref{hydrahead} (B)} 
When $\mm_0 = d\ZK$ the statement of (B) 
is just proposition~10 of~\cite{AFMY}, 
or indeed lemma~5.3 of~\cite{Kopp1}. 
When $\mm_0 = (d-3)\ZK$ we use the same argument. 
Briefly, we observe that by the definition of $d_\ell=d_\ell(D)$, 
$d_\ell-3 = u_D^\ell + u_D^{-\ell} - 2 = \frac{(u_D^\ell-1)^2}{u_D^\ell}$ 
and so $\mm_0 = (u_D^\ell-1)^2 \ZK$, which means that either 
$u_D^\ell \equiv 1 \bmod \mm_0$ or else that $(u_D^\ell-1)$ 
is nilpotent of index exactly~$2$. 
As it turns out, the only case for which 
$u_D^\ell \equiv 1 \bmod \mm_0$ is when $D=5$ 
and $\ell=1$: that is, when $d=4$, which is 
obviously a trivial case since $d-3=1$. 
For every other case, we have this neat degree-2 nilpotency. 

We now refer directly to the discussion 
around equations (15) and (16) in the proof of 
proposition~10 of~\cite{AFMY}. 
Again invoking lemma~\ref{Urk1}, 
$\U{(d-3)\ZK\{\}}$ must have a generator 
of the form $\pm u_D^k$ for some $k\geq1$. 
Hence in order to draw the same conclusions as we did there about 
the modulus~$d$, 
we must show that \emph{the unit kernel~$\U{(d-3)\ZK\{\}}$ contains no 
units which are totally negative}. 
This has the consequence that all four unit kernel 
groups~$\U{(d-3)\ZK\{\}}$, $\U{(d-3)\ZK\infty_1}$, 
$\U{(d-3)\ZK\infty_2}$ and $\U{(d-3)\ZK\infty_1\infty_2}$ are identical, 
and then by the formula~\eqref{Lang} quoted from 
Lang~\cite[Ch VI \S1 eq~(15)]{Lang} we 
see that the increases in the 
orders of the ray class groups are entirely 
governed by adding in successive real 
places to the modulus, as claimed.

So suppose to the contrary that we do in fact 
have a unit in~$\U{(d-3)\ZK\{\}}$ which 
is negative at both infinite places. 
In the cases where $\ZK$ has a negative fundamental unit $\uf$, 
the odd powers will always have 
mixed signs (one positive, the other negative). 
Hence from the discussion just above, any totally negative unit must be of 
the form~$-u_D^k$ for some non-zero $k \in\N$. 

It is an easy consequence of the basic properties 
of the modified Chebyshev polynomials $T_n^\ast(X)$---as defined in 
\cite{AFMY} or \cite{BGM}---that for every pair of positive integers $m,n$, 
the following rational integer congruence holds: 
\begin{equation}\label{dm3}
(d_{mn}-3) \equiv 0 \bmod (d_n-3) .
\end{equation}
We now need to be careful about the 
subscripts of the $d=d_j = d_j(D)$ for $j=1,k,\ell$. 
First of all, just running through the definitions implicit in the notation, 
since $d = d_\ell = d_\ell(D) = u_D^\ell + u_D^{-\ell} + 1$ for some $\ell$, 
the assumption that~$-u_D^k \in \U{(d_\ell-3)\ZK\{\}}$ is equivalent to  
$u_D^k \equiv -1 \bmod (d_\ell-3)\ZK$, which in turn implies 
\begin{equation}
d_k = u_D^k + u_D^{-k} + 1 \equiv -1 -1 + 1 \equiv -1 \bmod (d_\ell-3)\ZK . 
\end{equation}
But~\eqref{dm3} with $m=\ell$ and $n=1$ yields further that 
\begin{equation}\label{drm1}
d_k \equiv -1 \bmod (d_1-3)\ZK . 
\end{equation}
On the other hand, again from~\eqref{dm3} but with $m=k$ and $n=1$, we know 
that~$d_k-3$ is also divisible exactly by $d_1-3$, and so additionally 
\begin{equation}\label{drp3}
d_k \equiv 3 \bmod (d_1-3)\ZK . 
\end{equation}

But then~\eqref{drp3} together with~\eqref{drm1} says that 
$d_1-3$ divides into their difference, ie $3 - (-1) = 4$, inside $\ZK$. 
This leaves a very limited number of cases---that is to say, $d_1 = 4,5,7$---two  
of which ($4$ and $5$) we have excluded, and the other of 
which may be verified by direct calculation, since for all $k$ it 
is the case again by~\eqref{dm3} that 
$$
\U{(d_k-3)\ZK\{\}}\trianglelefteq \U{(d_1-3)\ZK\{\}} .
$$ 

So~$\U{(d-3)\ZK\{\}}$ has 
no totally negative units. 
\qed

\begin{remark}
We should also point out the yet more obvious fact that the same square-nilpotency 
phenomenon as becomes clear above in the proof of part (B) 
occurs for all moduli of the form~$(d_\ell+1)\ZK$, 
since $(u_D^\ell+1)\ZK$ is the square root of the ideal~$(d_\ell+1) \ZK$. 
\end{remark}


{\small

\begin{thebibliography}{99}

\bibitem{AYAZ} D. M. Appleby, H. Yadsan-Appleby, and G. Zauner, {\it Galois 
automorphisms of symmetric measurements}, Quant. Inf. Comp. {\bf 13} (2013) 672.

\bibitem{AFMY} M. Appleby, S. Flammia, G. McConnell, and J. Yard, 
{\it Generating ray class fields of real quadratic fields via complex 
equiangular lines}, Acta Arithmetica {\bf 192} (2020) 211; also published as 
arXiv:1604.06098.

\bibitem{Hilbert} D. Hilbert, {\it Matematische Probleme}, G\"ottinger Nachrichten 
(1900) 253; also published as {\it Mathematical problems}, Bull. AMS {\bf 8} 
(1902) 437. 

\bibitem{Stark1} H. M. Stark, {\it L-functions at s = 1. III. Totally real fields 
and Hilbert's twelfth problem}, Adv. Math. {\bf 22} (1976) 64. 

\bibitem{Zauner} G. Zauner: {\it Quantendesigns. Grundz\"uge einer 
nichtkommutativen Designtheorie}, PhD thesis, Univ. Wien 1999. Also published as  
{\it Quantum designs: Foundations of a noncommutative design theory}, 
Int. J. Quant. Inf. {\bf 9} (2011) 445.

\bibitem{Renes} J. M. Renes, R. Blume-Kohout, A. J. Scott, and C. M. Caves, 
{\it Symmetric informationally complete quantum measurements}, J. Math. Phys. {\bf 45} 
(2004) 2171. 

\bibitem{Chris} C. A. Fuchs, M. C. Hoang, and B. C. Stacey, {\it The SIC question: History 
and state of play}, Axioms {\bf 6} (2017) 21. 

\bibitem{Kopp1} G. S. Kopp, {\it SIC-POVMs and the Stark conjectures}, Int. Math. Res. Not. 
IMRN 2021.18 (2021) 13812. 

\bibitem{ABGHM} M. Appleby, I. Bengtsson, M. Grassl, M. Harrison, and G. McConnell, 
{\it SIC-POVMs from Stark units: Prime dimensions $n^2 + 3$}, J. Math. Phys. {\bf 63} 
(2022) 112205. 

\bibitem{BGM} I. Bengtsson, M. Grassl, and G. McConnell, {\it SIC-POVMs from Stark 
units: Dimensions $n^2 + 3 = 4p$, $p$ prime}, J. Math. Phys. {\bf 66} 
(2025) 082202. 

\bibitem{Kopp5} G.S. Kopp and J. C. Lagarias, {\it Ray class groups and ray class 
fields for orders of number fields}, Essential Number Theory {\bf 4} (2025) 1. 

\bibitem{Kopp2} G. S. Kopp and J. C. Lagarias, {\it SIC-POVMs and orders of real 
quadratic fields}, arXiv:2407.08048. 

\bibitem{Scott1} A. J. Scott and M. Grassl, {\it SIC-POVMs: A new computer study}, 
J. Math. Phys. {\bf 51} (2010) 042203.

\bibitem{Scott2} A. J. Scott, {\it SICs: Extending the list of solutions}, 
arXiv:1703.03993. 

\bibitem{Kopp3} G. S. Kopp, {\it The Shintani--Faddeev modular cocycle: Stark 
units from $q$-Pochhammer ratios}, arXiv:2411.06763.

\bibitem{Kopp4} M. Appleby, S. Flammia, and G. S. Kopp, {\it A constructive approach 
to Zauner's conjecture via the Stark conjectures}, arXiv:2501.03970.

\bibitem{Markuspower} M. Grassl, {\it On (square roots of) powers of Stark units 
in the fiducial vector}, unpublished manuscript (2023). 

\bibitem{Neukirch} J. Neukirch: {\it Algebraic Number Theory}, Springer, Berlin 2013.

\bibitem{Lang} S. Lang: {\it Algebraic Number Theory}, Springer, New York 1986. 

\bibitem{gras} Georges Gras, \emph{Class Field Theory: From Theory to Practice}, corrected second printing, Springer-Verlag Berlin (2005). 

\bibitem{Tate} J. Tate: {\it Les conjectures de Stark sur les fonctions $L$ d'Artin 
en $s=0$}, Progress in Mathematics {\bf 47}, Birkh\"auser, Boston 1984. 

\bibitem{Magma} W.~Bosma, J.~J. Cannon, and C.~Playoust, 
{\it The Magma algebra system I: The user language}, J. Symbolic Computation {\bf 24} (1997) 235. 

\bibitem{Marcus} D. M. Appleby, {\it SIC-POVMs and the extended Clifford group}, 
J. Math. Phys. {\bf 46} (2005) 052107. 

\bibitem{Fibonacci} M. Grassl and A. J. Scott, {\it Fibonacci--Lucas SIC-POVMs}, 
J. Math. Phys. {\bf 58} (2017) 122201.

\bibitem{monomial} D. M. Appleby, I. Bengtsson, S. Brierley, D. Gross, M. Grassl, 
and J.-\AA. Larsson, {\it The monomial representations of the Clifford group}, 
Quant. Inf. Comp. {\bf 12} (2012) 0404. 

\bibitem{BG} I. Bengtsson and M. Grassl, {\it A conjecture on almost flat SIC-POVMs}, 
arXiv:2512.13201. 

\bibitem{tbang} T. Bang, \emph{Congruence properties of Tchebycheff polynomials}, Mathematica Scandinavica \bf2\rm, (1954), 327--333.

\bibitem{nonprat} G. McConnell, {\it Some new infinite families of non-$p$-rational real quadratic fields}, arXiv:2406.14632. 

\bibitem{ACFW} M. Appleby, T.Y. Chien, S. Flammia, and S. Waldron, {\it Constructing 
exact symmetric informationally complete measurements from numerical solutions}, 
J. Phys. {\bf A51} (2018) 165302. 

\bibitem{Marcusray} M. Appleby, {\it Ray class SICs: Symmetries and Galois action}, 
unpublished manuscript (2018). 

\bibitem{Irina} M. Appleby, I. Bengtsson, I. Dumitru, and S. Flammia, {\it Dimension 
towers of SICs. I. Aligned SICs and embedded tight frames}, J. Math. Phys. {\bf 58} (2017) 
112201. 

\bibitem{Basudha} I. Bengtsson and B. Srivastava, {\it Dimension towers of SICs. II. 
Some constructions}, J. Phys. {\bf A55} (2022) 215302.

\bibitem{Ole} O. Andersson and I. Dumitru, {\it Aligned SICs and embedded tight frames in 
even dimensions}, J. Phys. {\bf A52} (2019) 4525302. 

\bibitem{Danylo} V. Ostrovskyi and D. Yakymenko, {\it Geometric properties of SIC-POVM 
tensor square}, Lett. Math. Phys. {\bf 112} (2022) 7. 

\bibitem{cohenstev} H. Cohen and P. Stevenhagen, \emph{Computational class field theory}, in~\emph{Algorithmic Number Theory: Lattices, Number Fields, Curves and Cryptography}, edited by Joseph P.\ Buhler and Peter Stevenhagen. MSRI Publications Volume \textbf{44} (2008), 497-534. 


\bibitem{gelman} S. I. Gelfand and Y. J. Manin, \emph{Methods of Homological Algebra}, SMM Springer Verlag 2nd Edition (2003). 

\bibitem{Weyl} H. Weyl: {\it Gruppentheorie und Quantenmechanik}, Hirzel, 
Leipzig 1928; also published as {\it Theory of Groups and Quantum Mechanics}, 
Dutton, New York 1932.

\bibitem{Schwinger} J. Schwinger, {\it Unitary operator bases}, Proc. Natl. Acad. Sci. 
{\bf 46} (1960) 570. 

\bibitem{Lane} L. P. Hughston and S. M. Salamon, {\it Surveying points in the 
complex projective plane}, Adv. Math. {\bf 286} (2016) 1017. 

\bibitem{Feri} F. Sz{\"o}ll{\H{o}}si, {\it All complex equiangular tight frames 
in dimension 3}, arXiv:1402.6429.

\bibitem{Samuel} S. B. Samuel and Z. Gedik, {\it Group theoretical classification 
of SIC-POVMs}, J. Phys. {\bf A57} (2024) 295304. 

\bibitem{Blake} B. C. Stacey: {\it A First Course in the Sporadic SICs}, Springer 2021. 


\end{thebibliography}

}
\end{document}